\documentclass[conference]{IEEEtran}
\IEEEoverridecommandlockouts
\usepackage{cite}
\usepackage{amsmath,amssymb,amsfonts}
\usepackage{algorithmic}
\usepackage{graphicx}
\usepackage{textcomp}
\usepackage{xcolor}
\usepackage{amsthm}
\usepackage{amsmath}
\usepackage{verbatim}
\usepackage{algorithmic}
\usepackage[linesnumbered,ruled,vlined]{algorithm2e}
\usepackage{subfigure}
\usepackage{amsmath}
\usepackage{float}
\usepackage{url}
\newtheorem{theorem}{Theorem}
\newtheorem{lemma}{Lemma}
\newtheorem{definition}{Definition}

\def\BibTeX{{\rm B\kern-.05em{\sc i\kern-.025em b}\kern-.08em
    T\kern-.1667em\lower.7ex\hbox{E}\kern-.125emX}}
\begin{document}

\title{ProMIPS: Efficient High-Dimensional $c$-Approximate Maximum Inner Product Search with a Lightweight Index
\thanks{* Corresponding author (email: guyu@mail.neu.edu.cn)}
}
\author{\IEEEauthorblockN{Yang Song}
\IEEEauthorblockA{\textit{School of Computer Science and Engineering} \\
\textit{Northeastern University}\\
Shenyang, China \\
ysqyw1994@163.com}
\and
\IEEEauthorblockN{Yu Gu$^{*}$}
\quad\quad\quad\quad\quad\quad\quad\quad\quad\quad\quad\quad\quad\quad\quad\quad\quad\quad\quad\quad\quad\IEEEauthorblockA{\textit{School of Computer Science and Engineering} \\
\textit{Northeastern University}\\
Shenyang, China \\
guyu@mail.neu.edu.cn}
\and
\IEEEauthorblockN{Rui Zhang}
\quad\quad\quad\quad\quad\quad\quad\quad\quad\quad\quad\quad\quad\quad\quad\quad\quad\quad\quad\quad\quad\quad\quad\quad\quad\quad\IEEEauthorblockA{\url{www.ruizhang.info} \\
rui.zhang@ieee.org}
\and
\IEEEauthorblockN{Ge Yu}
\IEEEauthorblockA{\textit{School of Computer Science and Engineering} \\
\textit{Northeastern University}\\
Shenyang, China \\
yuge@mail.neu.edu.cn}
}
\maketitle
\begin{abstract}
Due to the wide applications in recommendation systems, multi-class label prediction and deep learning, the Maximum Inner Product (MIP) search problem has received extensive attention in recent years. Faced with large-scale datasets containing high-dimensional feature vectors, the state-of-the-art LSH-based methods usually require a large number of hash tables or long hash codes to ensure the searching quality, which takes up lots of index space and causes excessive disk page accesses. In this paper, we relax the guarantee of accuracy for efficiency and propose an efficient method for $c$-Approximate Maximum Inner Product ($c$-AMIP) search with a lightweight iDistance index. We project high-dimensional points to low-dimensional ones via 2-stable random projections and derive probability-guaranteed searching conditions, by which the $c$-AMIP results can be guaranteed in accuracy with arbitrary probabilities. To further improve the efficiency, we propose Quick-Probe for quickly determining the searching bound satisfying the derived condition in advance, avoiding the inefficient incremental searching process. Extensive experimental evaluations on four real datasets demonstrate that our method requires less pre-processing cost including index size and pre-processing time. In addition, compared to the state-of-the-art benchmark methods, it provides superior results on searching quality in terms of overall ratio and recall, and efficiency in terms of page access and running time.
\end{abstract}


\maketitle
\section{Introduction}\label{sec:introduction}
Given a dataset $D$ of $n$ data points and a query point $q$ in $d$-dimensional space $R^{d}$, a Maximum Inner Product (MIP) search returns the point $o^{*}\in D$ maximizing the inner product with $q$. Mathematically, it is represented as $o^{*}=\arg \max_{o\in D}\langle o,q\rangle$.
The so-called MIP search is a fundamental problem and it has been widely applied in various domain areas, such as matrix factorization based recommendation systems~\cite{DBLP:conf/sigmod/LiCYM17, DBLP:conf/icdm/LiuCTWZ15, DBLP:conf/recsys/BachrachFGKKNP14, DBLP:conf/recsys/CremonesiKT10}, multi-class label prediction~\cite{DBLP:conf/cvpr/DeanRSSVY13} and deep learning~\cite{DBLP:conf/kdd/SpringS17}. Typically, in matrix factorization based recommendation systems, the vectors $q$ and $o$ are viewed as latent features for a user and a product, respectively. The inner product between $q$ and $o$ reflects the user's interest in the product. Therefore, MIP search is an important concern in these recommendation systems to recommend popular products to users.

The phenomenon of the ``Dimensionality Curse" makes exact MIP search in high-dimensional space very expensive. Therefore, many researchers set their sights on the approximate version of the MIP search problem~\cite{DBLP:journals/corr/AuvolatV15, DBLP:conf/recsys/BachrachFGKKNP14, DBLP:conf/aistats/GuoKCS16, DBLP:conf/icml/NeyshaburS15, DBLP:conf/nips/Shrivastava014, DBLP:conf/uai/Shrivastava015, DBLP:journals/corr/VijayanarasimhanSMY14, DBLP:conf/kdd/HuangMFFT18, DBLP:conf/nips/YanLDCC18}, which is called $c$-Approximate MIP ($c$-AMIP) search problem. Mathematically, given an approximation ratio $c$ ($0<c<1$) and a query point $q$, $c$-AMIP search returns a point $o\in D$ such that $\langle o,q\rangle\geq c\langle o^{*},q\rangle$, where $o^{*}$ is the exact MIP point of $q$. In this way, a good accuracy-efficiency trade-off can be provided where the efficiency can be improved significantly while only a small amount of errors occur.

At present, the state-of-the-art methods for $c$-AMIP search are transformation-based. In these methods, a MIP search is converted into a Nearest Neighbor (NN) search or a Maximum Cosine-similarity (MC) search by transforming the given data points and the query point asymmetrically or symmetrically, and employ Locality-Sensitive Hashing (LSH) to solve the NN or MC search problem. These LSH-based methods improve the searching efficiency, but to achieve satisfactory accuracy, they require more hash vectors to project high-dimensional points onto more hash values, indexed by heavyweight structures in terms of massive hash tables. These heavyweight structures require more maintenance overhead, which increases linearly as the number of hash tables increases. Especially in commonly used mobile devices or IoT devices, a huge amount of data will be frequently inserted or deleted in a short time, where the heavyweight index requiring more maintenance overhead may cause delays. Besides, hundreds or thousands of hash tables may also lead to more disk page accesses which degrades the efficiency when storing data points on disks.

Motivated by these existing restrictions, we attempt to design an efficient method for $c$-AMIP search with a lightweight index. A recent method, SRS~\cite{DBLP:journals/pvldb/SunWQZL14}, which can be considered as a special version of LSH technique, projects high-dimensional points onto low-dimensional ones via 2-stable random projections to reduce high-dimensional $c$-ANN search to low-dimensional NN search, and perform the low-dimensional search through a lightweight index in terms of R-tree. Compared to the standard LSH, SRS can directly project high-dimensional points onto lower-dimensional ones with fewer projections, avoiding the heavyweight index. Although it's designed for Euclidean distance, it presents a new angle to solve $c$-AMIP search problem since the Euclidean distance between two points can be computed by their inner product and 2-norms. Even though, it's still challenging to follow the direction of SRS to solve $c$-AMIP search problem. Since inner product isn't a metric measurement, some basic necessary properties such as non-negativity and triangle inequality are not satisfied. Without these properties, we can't derive the probability-guaranteed searching conditions for $c$-AMIP search directly like SRS.

Inspired by SRS, we also project high-dimensional points onto low-dimensional ones via 2-stable random projections. Based on the projection and the properties of inner product, we theoretically derive two conditions specifically for $c$-AMIP search. According to the conditions, we perform an incremental NN search in low-dimensional space to collect the candidate points until a point satisfying either of the conditions is searched. And the required $c$-AMIP point is guaranteed to appear among these candidate points with the given probability. However, during the incremental NN search, every time a point is returned, it is required to determine whether it satisfies the condition, which is a time-consuming procedure. To avoid the procedure, we come up with a quick method named Quick-Probe for directly locating the point satisfying the searching condition to determine the searching range, which enables us to replace the incremental NN search with a range search without testing each returned NN point. Meanwhile, based on Quick-Probe, we can also compute an optimized projected dimension to pursue a more efficient searching process. With respect to the index used for search, since the optimized dimensions are usually greater than 3, R-tree used in SRS isn't applicable. Hence, in order to search in higher-dimensional space, we adopt iDistance~\cite{DBLP:journals/tods/JagadishOTYZ05}, which is an efficient index, and design a new partition pattern for it. Compared to LSH-based methods, iDistance is a typical lightweight index, which only requires a single B+-tree to orderly organize points on disks, rather than a large number of hash tables or long hash codes.

As can be seen from the above descriptions, we propose an efficient method for the probability-guaranteed high-dimensional $c$-AMIP search with a lightweight index. Our contributions are summarized as follows:
\begin{itemize}
  \item We employ 2-stable random projections to project high-dimensional points onto low-dimensional points and theoretically derive two searching conditions for $c$-AMIP search. Relying on the conditions, the $c$-AMIP result can be guaranteed in accuracy with arbitrary probabilities.
  \item Quick-Probe is proposed for quickly locating the point to determine the searching range, which avoids testing each returned point repeatedly to accelerate the searching process. Besides, an optimized projected dimension can be computed based on Quick-Probe.
  \item Extensive experimental evaluations on four real datasets show that our method occupies a smaller index size and requires less pre-processing time compared to benchmark methods. Furthermore, our method is also superior in accuracy measured by overall ratio and recall, and efficiency measured by page access and running time.
\end{itemize}

The rest of the paper is structured as follows. Section~\ref{sec:preliminaries} presents the preliminaries. We introduce the overall framework in Section~\ref{sec:overallframework}. The searching conditions are presented in Section~\ref{sec:searchingconditions}. We propose Quick-Probe in Section~\ref{sec:quickprobe}. In Section~\ref{sec:index}, we describe the indexing technique. The time and space complexities are theoretically analyzed in Section~\ref{sec:complexities}. Experimental evaluations are discussed in Section~\ref{sec:experimentalevaluations}. The related works are introduced in Section~\ref{sec:relatedwork}. Finally, we conclude our work in Section~\ref{sec:conclusion}.

\section{Preliminaries}\label{sec:preliminaries}
\subsection{Problem Definition}
Given a dataset $D$ containing $n$ data points in a $d$-dimensional space $R^{d}$, the inner product $\langle o,q\rangle$ between two points $o=(o_{1},o_{2},...,o_{d})$ and $q=(q_{1},q_{2},...,q_{d})$ can be computed as $\langle o,q\rangle=\Sigma_{i=1}^{d} o_{i}q_{i}$. Inner product is widely used in real applications where the MIP search problem plays an important role. For example, in recommender systems, $o$ and $q$ are used as a user vector and an item vector, respectively. A higher inner product between $o$ and $q$ indicates that the item better suits the user's preference~\cite{DBLP:conf/recsys/BachrachFGKKNP14}.

In this paper, to handle the high-dimensional cases, we allow a trade-off between accuracy and efficiency, and focus on $c$-AMIP search problem formally defined as follows:
\begin{definition}[$c$-AMIP search problem]\label{Definition 1}
Given a query point $q\in R^{d}$ and an approximation ratio $c$ ($0<c<1$), $c$-AMIP search is to find a point $o\in D$ such that $\langle o,q\rangle\geq c\langle o^{*},q\rangle$, where $o^{*}$ is the $q$'s exact MIP point in $D$.
\end{definition}

Similarly, $c$-$k$-AMIP search is to find $k$ points $o_{i}\in D$ ($1\leq i\leq k$) such that $\langle o_{i},q\rangle\geq c\langle o_{i}^{*},q\rangle$, where $o_{i}^{*}$ is the $i^{th}$ exact MIP point of $q$ in $D$.

\subsection{2-Stable Random Projection}
\begin{definition}[2-Stable Random Projections]\label{Definition2}
Given a $d$-dimensional point $o$, which can be considered as a vector $\overrightarrow{o}$, and a $d$-dimensional vector $\overrightarrow{v}$, whose entries are i.i.d. random variables following the standard normal distribution $N(0,1)$, 2-Stable Random Projections is to compute $f(o)=\overrightarrow{v}\cdot \overrightarrow{o}$.
\end{definition}

Based on 2-stable random projections, we can obtain the following Lemma~\cite{DBLP:conf/soda/Panigrahy06}.

\begin{lemma}\label{Lemma1}
For any $o_{1},o_{2}\in R^{d}$, $f(o_{1})-f(o_{2})$ follows the normal distribution $N(0,dis^{2}(o_{1},o_{2}))$.
\end{lemma}

In our method, the projected dimension of each point is $m$. Therefore, we perform $m$ 2-stable random projections to obtain $m$-dimensional projected points.

\subsection{IDistance}
IDistance is an efficient index based on B+-tree for the exact similarity search~\cite{DBLP:journals/tods/JagadishOTYZ05}, which is illustrated in Fig.~\ref{IDistance}. In iDistance, the whole indexing space is divided into several partitions centered at their reference points. In these partitions, points are transformed into a single dimensional value based on their distances to their corresponding reference points and these values are indexed by a B+-tree. Based on the B+-tree, similarity search can be performed. For example, in Fig.~\ref{IDistance}, given a query point and a searching radius, the grey area in the B+-tree will be searched, so that the points in the gray areas of the space are fetched for determining the final searching results. In this paper, we utilize iDistance as the index to accelerate the searching process.

\begin{figure}\centering
\includegraphics[width=0.45\textwidth]{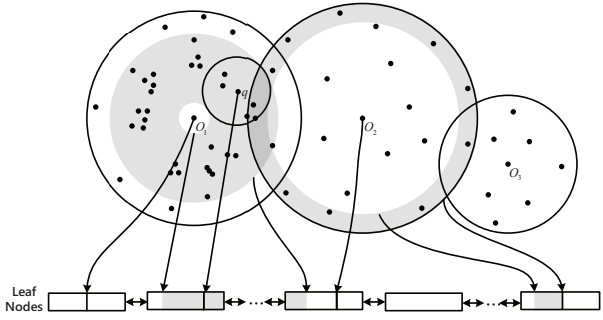}
\caption{IDistance}\label{IDistance}
\vspace{-15pt}
\end{figure}

We summarize the frequently-used symbols in Table~\ref{Notations}.
\begin{table}\centering
\caption{Frequently used symbols}
\label{Notations}
\begin{tabular}{ll}
\hline\noalign{\smallskip}
Symbol & Explanation \\
\noalign{\smallskip}\hline\noalign{\smallskip}
$D$ & dataset \\
$n$ & number of data points \\
$o$, $q$ & data point, query point \\
$P(o)$, $P(q)$ & projected data point, projected query point \\
$d$ & original dimensionality of each point \\
$m$ & projected dimensionality of each point \\
$o^{*}$, $o_{i}^{*}$ & the MIP point, the $i$-th MIP point of $q$ \\
$N(a,b)$ & the normal distribution with mean $a$ \\
         & and variance $b$ \\
$dis(o_{1},o_{2})$ & the Euclidean distance between $o_{1}$ and $o_{2}$ \\
$\|o\|$ & the norm of point $o$ \\
$o_{M}$ & the point with the maximum norm \\
        & in the original space \\
$\langle o_{1},o_{2}\rangle$ & the inner product between $o_{1}$ and $o_{2}$ \\
$\chi^{2}(m)$ & the chi-square distribution with $m$ degrees \\
              & of freedom \\
$\Psi_{m}(x)$ & cumulative distribution function of $\chi^{2}(m)$ \\
$\Psi_{m}^{-1}(p)$ & inverse function of $\Psi_{m}(x)$ \\
$k$ & number of the returned points \\
$c$ & approximation ratio \\
$p$ & guaranteed probability \\
\noalign{\smallskip}\hline
\end{tabular}
\vspace{-20pt}
\end{table}

\section{Overall Framework}\label{sec:overallframework}
In this paper, to solve the probability-guaranteed $c$-AMIP search problem in high-dimensional space, we project high-dimensional points onto low-dimensional ones via 2-stable random projections. Since the ratio of points' Euclidean distance in high-dimensional space and low-dimensional space follows the chi-square distribution, and points' Euclidean distance is related to their inner product, we can derive two searching conditions. Based on these conditions, we perform incremental NN search in low-dimensional space for the probability-guaranteed $c$-AMIP point. In detail, every time a point is returned, we test if the point satisfies either of the conditions to determine whether to terminate the incremental NN search. If satisfied, the $c$-AMIP point exists in the searched points with the given probability at least. The searching conditions are elaborated in Section~\ref{sec:searchingconditions}.

However, it is time-consuming to perform the incremental NN search and test each returned point one by one. To avoid it, we propose a method to quickly locate the point satisfying the searching condition, called Quick-Probe. The method quickly locates the point through binary transformation and data norm's properties. In this way, the searching range is directly determined by the located point and we can perform range search instead of the incremental NN search elaborated in Section~\ref{sec:searchingconditions}, to collect the candidate points. Benefiting from Quick-Probe, we no longer do any incremental NN search and the time-consuming process of testing the returned point one by one can be avoided. Quick-Probe is elaborated in Section~\ref{sec:quickprobe}. In addition, we adopt iDistance as the index and design a new partition pattern for it for performing searching tasks more efficiently in low-dimensional space, which is elaborated in Section~\ref{sec:index}.

Based on the searching conditions and Quick-Probe, our method is described in two parts including the pre-process and the searching process. In the pre-process, the original high-dimensional points are projected onto projected low-dimensional ones. In the low-dimensional space, the index structure is constructed for performing the searching tasks and the low-dimensional points and their corresponding high-dimensional ones are organized on disks. In addition, the projected points are also converted into binary codes and each point's norms are computed, for determining the searching range according to Quick-Probe. In the searching process, Quick-Probe is applied to find the point satisfying the condition and determine the searching range, by which the range search is performed in the projected space to find the candidate points. These candidate points are verified using their inner products in the original space for returning the $c$-AMIP search results. To clearly summarize our method's overall framework, we give Fig.~\ref{overview} to describe it.

\begin{figure}\centering
\includegraphics[width=0.5\textwidth]{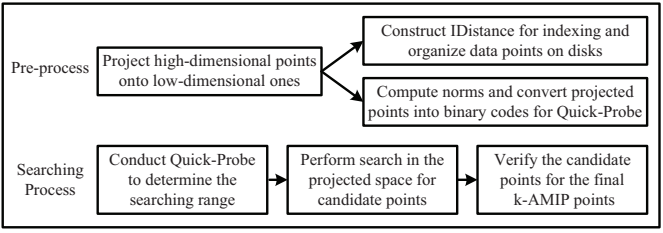}
\caption{Overall Framework} \label{overview}
\vspace{-15pt}
\end{figure}

\section{Searching Conditions}\label{sec:searchingconditions}
Our method aims to guarantee the $c$-AMIP search in accuracy with arbitrary probabilities by proposing two searching conditions. In this section, we will introduce the conditions and prove their validity.

\subsection{Condition A}
As stated above, we perform incremental NN search in the projected space. During the searching process, we fetch every returned point as the candidate points. If the current returned $P(q)$'s $i$-th NN point $P(o_{i})$ satisfies:
\begin{equation}\label{Condition1}
{\|o_{M}\|}^{2}+{\|q\|}^{2}-\frac{2{\langle o_{i},q\rangle}}{c}\leq 0,
\end{equation}
a $c$-AMIP point must exist among these candidate points, and the searching process can be terminated, where $o$ and $q$ are the corresponding original points of $P(o)$ and $P(q)$, $o_{M}$ is the point with the maximum norm in the original space. Formula~\ref{Condition1} is considered as Condition A and the following Theorem~\ref{Theorem1} proves its validity.

\begin{theorem}\label{Theorem1}
If the current returned NN point satisfies Formula~\ref{Condition1}, a $c$-AMIP point must exist among the points having been returned.
\end{theorem}
\begin{proof}
We assume that $o^{*}$ is the exact MIP point of the query point $q$. If ${\|o_{M}\|}^{2}+{\|q\|}^{2}-\frac{2{\langle o_{i},q\rangle}}{c}\leq 0$, since $o_{M}$ is the point with the maximum norm, we have ${\|o^{*}\|}^{2}+{\|q\|}^{2}-\frac{2{\langle o_{i},q\rangle}}{c}\leq 0$.

Since ${\|o^{*}\|}^{2}+{\|q\|}^{2}-2{\langle o^{*},q\rangle}\geq 0$, we have $\langle o_{i},q\rangle\geq c\langle o^{*},q\rangle$.
Therefore, when ${\|o_{M}\|}^{2}+{\|q\|}^{2}-\frac{2{\langle o_{i},q\rangle}}{c}\leq 0$, a $c$-AMIP point must have been accessed when $o_{i}$ is searched.
\end{proof}

\subsection{Condition B}
During the incremental NN search in the projected space, if the current returned NN point doesn't satisfy Condition A, ${\|o_{M}\|}^{2}+{\|q\|}^{2}-\frac{2{\langle o_{i},q\rangle}}{c}>0$ is true. Based on it, we turn to test the returned NN point by the following Formula~\ref{Condition2}. If $P(q)$'s $i$-th NN point $P(o_{i})$ satisfies:
\begin{equation}\label{Condition2}
\Psi_{m}(\frac{dis^{2}(P(o_{i}),P(q))}{{\|o_{M}\|}^{2}+{\|q\|}^{2}-\frac{2{\langle o_{max},q\rangle}}{c}})\geq p,
\end{equation}
a $c$-AMIP point must exist among these candidate points with the given probability $p$ at least, and the searching process can be terminated, where $o_{max}$ is the point with the maximum inner product to $q$ among all the candidate points having been returned so far. Formula~\ref{Condition2} is considered as Condition B and the following Theorem~\ref{Theorem2} proves its validity.

Before proving Theorem~\ref{Theorem2}, we firstly give the following Lemma~\ref{Lemma2} as the preparations.

\begin{lemma}\label{Lemma2}
$\frac{dis^{2}(P(o),P(q))}{{\|o\|}^{2}+{\|q\|}^{2}-2{\langle o,q\rangle}}$ follows the $\chi^{2}(m)$ distribution.
\end{lemma}
\begin{proof}
According to Definition~\ref{Definition2}, we select $m$ $d$-dimensional vectors  whose entries are i.i.d random variables following $N(0,1)$ for performing $m$ 2-stable random projections to get $m$-dimensional projected points. The $m$-dimensional projected points are denoted as $P(o)=(f_{1}(o),f_{2}(o),...,f_{m}(o))$ and $P(q)=(f_{1}(q),f_{2}(q),...,f_{m}(q))$.

According to Lemma~\ref{Lemma1}, we have $\frac{f_{i}(o)-f_{i}(q)}{dis(o,q)}\sim N(0,1)$ $(1\leq i\leq m)$.

Therefore, we can obtain $\sum_{i=1}^{m}{(\frac{f_{i}(o)-f_{i}(q)}{dis(o,q)})}^{2}\sim \chi^{2}(m)$.

Since $dis^{2}(P(o),P(q))=\sum_{i=1}^{m} ({f_{i}(o)-f_{i}(q)})^{2}$ and $dis^{2}(o,q)={\|o\|}^{2}+{\|q\|}^{2}-2{\langle o,q\rangle}$, the lemma can be proved.
\end{proof}

\begin{theorem}\label{Theorem2}
If the current returned NN point satisfies Formula~\ref{Condition2}, a $c$-AMIP point must exist among the points having been returned with probability $p$ at least.
\end{theorem}
\begin{proof}
Assume that $o^{*}$ is the exact MIP point, we consider the relationship between $dis(P(o^{*}),P(q))$ and $dis(P(o_{i}),P(q))$. We discuss their relationship in two cases:
\begin{itemize}
  \item \textbf{C1:} $dis(P(o^{*}),P(q))\leq dis(P(o_{i}),P(q))$.

  In this case, since we perform the incremental NN search, $o^{*}$ must have been accessed when $o_{i}$ is searched.

  \item \textbf{C2:} $dis(P(o^{*}),P(q))> dis(P(o_{i}),P(q))$.

  In this case, our method may produce incorrect results if none of the $c$-AMIP points has appeared so far, which can also be represented as $\langle o_{max},q\rangle<c\cdot \langle o^{*},q\rangle$. However, we can prove that the probability of such incorrect case is less than $1-p$.

  According to Lemma~\ref{Lemma2}, for any $x>0$ and $o$, we have
  $$Pr[dis(P(o),P(q))\leq x]=\Psi_{m}(\frac{x^{2}}{{\|o\|}^{2}+{\|q\|}^{2}-2{\langle o,q\rangle}}).$$

  Based on it, we have:
  \begin{align*}
  &\quad Pr[dis(P(o^{*}),P(q))> dis(P(o_{i}),P(q))]\\
  &=1-\Psi_{m}(\frac{dis^{2}(P(o_{i}),P(q))}{\|o^{*}\|^{2}+\|q\|^{2}-2\langle o^{*},q\rangle}).
  \end{align*}
  Since $\langle o_{max},q\rangle<c\cdot \langle o^{*},q\rangle$, we can derive
  \begin{align*}
  &\quad \Psi_{m}(\frac{dis^{2}(P(o_{i}),P(q))}{\|o^{*}\|^{2}+\|q\|^{2}-2\langle o^{*},q\rangle})\\
  &>\Psi_{m}(\frac{dis^{2}(P(o_{i}),P(q))}{\|o^{*}\|^{2}+\|q\|^{2}-\frac{2\langle o_{max},q\rangle}{c}}).
  \end{align*}

  Since $o_{M}$ is the point with the maximum norm, we have
  \begin{align*}
  &\quad \Psi_{m}(\frac{dis^{2}(P(o_{i}),P(q))}{\|o^{*}\|^{2}+\|q\|^{2}-\frac{2\langle o_{max},q\rangle}{c}})\\
  &>\Psi_{m}(\frac{dis^{2}(P(o_{i}),P(q))}{\|o_{M}\|^{2}+\|q\|^{2}-\frac{2\langle o_{max},q\rangle}{c}}).
  \end{align*}

  Therefore, if $o_{i}$ satisfies $\Psi_{m}(\frac{dis^{2}(P(o_{i}),P(q))}{{\|o_{M}\|}^{2}+{\|q\|}^{2}-\frac{2{\langle o_{max},q\rangle}}{c}})\geq p$, we have $Pr[dis(P(o^{*}),P(q))> dis(P(o_{i}),P(q))]\leq 1-p$.
\end{itemize}
\end{proof}

Algorithm~\ref{MIPSearch} gives the pseudo-code of the searching process. The algorithm can also be extended to solve the $c$-$k$-MIP search problem by some simple changes. In Condition A, it's required to test the current $k$-th MIP point $o_{max}^{k}$. Similarly, we should use $o_{max}^{k}$ in Condition B instead of $o_{max}$.

\begin{algorithm}
\caption{\textsf{MIP-Search-\uppercase\expandafter{\romannumeral1}} ($D,n,c,p,q$)}\label{MIPSearch}
\LinesNumbered
$o_{max}\leftarrow Null$;\\
$i\leftarrow 1$;\\
// Perform incremental NN search\\
\While{$i\leq n$}{
$P(o_{i})\leftarrow$ $P(q)$'s $i$-NN point;\\
\If{$\langle o_{max},q\rangle\leq \langle o_{i},q\rangle$}{
$o_{max}\leftarrow o_{i}$; // Update MIP point\\
}

\If{Condition A}{
\Return $o_{max}$;
}
\ElseIf{Condition B}{
\Return $o_{max}$;
}
$i\leftarrow i+1$;\\
}
\Return $o_{max}$;
\end{algorithm}
\vspace{-12pt}
\section{Quick-Probe}\label{sec:quickprobe}
As can be seen from Algorithm~\ref{MIPSearch}, we have to perform the incremental NN search to find the point satisfying the searching condition. Whenever a point is returned, it's required to test it using Condition A or Condition B. Especially in Condition B, Euclidean distance in the projected space is computed, which is time-consuming when the projected dimension is high. Besides, it also incurs extra page accesses when fetching points from disks. Therefore, we attempt to avoid testing the points one by one.

\subsection{Method}\label{subsec:method}
For the purpose, we introduce a method named Quick-Probe, by which we can quickly locate a point satisfying Condition B as much as possible. The distance between the point and the query in the projected space is used as an estimation of the searching range. It enables us to perform range search instead of incremental NN search to find the candidates points within the searching range.

Nevertheless, it's hard to determine the bound of $\frac{dis^{2}(P(o),P(q))}{\|o_{M}\|^{2}+\|q\|^{2}-\frac{2\langle o_{max},q\rangle}{c}}$ in Condition B and locate a point satisfying the condition. But we observe that the point satisfying Formula~\ref{Condition2a} is more likely to satisfy Condition B and the determined searching range can infinitely approach the range determined by Condition B. So we turn our attention to determine the bound of $\frac{dis^{2}(P(o),P(q))}{c\times dis^{2}(o,q)}$, and locate a point satisfying Formula~\ref{Condition2a}.
\begin{equation}\label{Condition2a}
\Psi_{m}(\frac{dis^{2}(P(o),P(q))}{c\times dis^{2}(o,q)})\geq p
\end{equation}

The bound is determined through binary transformation and data norm's properties. In detail, we transform each projected point into a binary code $c(o)=(c_{1}(o),c_{2}(o),...,c_{m}(o))$, where $c_{i}(o)=1$ if $P_{i}(o)$ is non-negative and $c_{i}(o)=0$ otherwise ($i=1,2,...,m$). According to Theorem~\ref{Theorem3}, we can derive the lower bound of $dis(P(o),P(q))$. The upper bound of $dis(o,q)$ can be derived through Theorem~\ref{Theorem4} using the property of data norm. By Theorem~\ref{Theorem3} and Theorem~\ref{Theorem4}, a lower bound of $\frac{dis^{2}(P(o),P(q))}{c\times dis^{2}(o,q)}$ can be computed. If the lower bound referring to a point $o$ is greater than $\Psi_{m}^{-1}(p)$, Formula~\ref{Condition2a} must be satisfied. Therefore, we can use $dis(P(o),P(q))$ as the searching range in the projected space. The process of finding $o$ is described as below.

\begin{theorem}\label{Theorem3}
The lower bound of the Euclidean distance between $P(o)$ and $P(q)$ is $\frac{1}{\sqrt{m}}\sum_{i=1}^{m}(c_{i}(o)\oplus c_{i}(q))\times |P_{i}(q)|$.
\end{theorem}
\begin{proof}
For any $m$-dimensional vector $x$, it holds that $\sqrt{m}\|x\|_{2}\geq \|x\|_{1}$~\cite{zhang2011matrix, DBLP:conf/sigmod/LiYZXCLNC18}. Therefore, we have $\|P(o)-P(q)\|_{2}\geq \frac{1}{\sqrt{m}}\|P(o)-P(q)\|_{1}$. When $c_{i}(o)=c_{i}(q)$, $P_{i}(o)$ and $P_{i}(q)$ have the same sign and $c_{i}(o)\oplus c_{i}(q)=0$ holds. Since $|P_{i}(o)-P_{i}(q)|\geq 0$, we have $|P_{i}(o)-P_{i}(q)|\geq (c_{i}(o)\oplus c_{i}(q))\times |P_{i}(q)|$. When $c_{i}(o)\neq c_{i}(q)$, $P_{i}(o)$ and $P_{i}(q)$ have different signs and $c_{i}(o)\oplus c_{i}(q)=1$ holds. Therefore, we also have $|P_{i}(o)-P_{i}(q)|=|P_{i}(o)|+|P_{i}(q)|\geq (c_{i}(o)\oplus c_{i}(q))\times |P_{i}(q)|$. Therefore, it holds that $|P_{i}(o)-P_{i}(q)|\geq (c_{i}(o)\oplus c_{i}(q))\times |P_{i}(q)|$ and we can obtain that
\begin{equation}\label{equation4}
\begin{split}
\|P(o)-P(q)\|_{2} &\geq \frac{1}{\sqrt{m}}\sum_{i=1}^{m}(c_{i}(o)\oplus c_{i}(q))\times |P_{i}(q)|
\end{split}
\end{equation}
\end{proof}

\begin{theorem}\label{Theorem4}
The upper bound of the Euclidean distance between $o$ and $q$ is $\sum_{i=1}^{m}|o_{i}|+\sum_{i=1}^{m}|q_{i}|$.
\end{theorem}
\begin{proof}
According to the property of vector norm and absolute value equality~\cite{zhang2011matrix}, we can simply derive:
\begin{equation}\label{equation5}
\|o-q\|_{2} \leq \|o-q\|_{1} \leq \sum_{i=1}^{m}|o_{i}|+\sum_{i=1}^{m}|q_{i}|=\|o\|_{1}+\|q\|_{1}.
\end{equation}
\end{proof}

In the pre-process, the projected points with the same binary code will be divided into the same group, and the 1-norms of their original points are computed and sorted. In the searching process, the lower bounds of Euclidean distance between each group and the query point are computed through Formula~\ref{equation4}. We search the groups in ascending order of their lower bounds. In each group, its lower bound is denoted as $LB$ and we fetch the point $o$ whose $\|o\|_{1}$ is the smallest among the points in the group to find the largest value of $\frac{LB^{2}}{c\times (\|o\|_{1}+\|q\|_{1})^{2}}$. Then we test whether it satisfies $\Psi_{m}(\frac{LB^{2}}{c\times (\|o\|_{1}+\|q\|_{1})^{2}})\geq p$, which is denoted as Test A. If Test A is satisfied, we fetch the point to determine the searching range. If not satisfied, we record the point's value of $\frac{LB^{2}}{c\times (\|o\|_{1}+\|q\|_{1})^{2}}$ and continue to search in the next group until the point is found. If there is no point satisfying it in all groups, we fetch the point with the largest recorded value of $\frac{LB^{2}}{c\times (\|o\|_{1}+\|q\|_{1})^{2}}$ as the result. The following Algorithm~\ref{QP} describes the whole process. In Algorithm~\ref{QP}, $G=\{G_{1},G_{2},...,G_{K}\}$ are the input set of groups with the same binary codes. In each group, the points $o$ are sorted in the ascending order of $\|o\|_{1}$.

\begin{algorithm}
\caption{\textsf{Quick-Probe} ($G,c,p,q$)}\label{QP}
\LinesNumbered
Compute each group $G_{i}$'s lower bound $LB_{i}$;\\
$\{GS_{1},GS_{2},...,GS_{K}\}\leftarrow$ the sorted groups in the ascending order of the lower bounds;\\
$point\leftarrow Null$;\\
$value\leftarrow 0$;\\
\For{$i=1$ to $K$}{// Test A\\
\If{$\Psi_{m}(\frac{LB_{i}^{2}}{c\times (\|o_{i1}\|_{1}+\|q\|_{1})^{2}})\geq p$}{
\Return $o_{i1}$;
}

// Update the point with the largest value\\
\If{$\frac{LB_{i}^{2}}{c\times (\|o_{i1}\|_{1}+\|q\|_{1})^{2}}\geq value$}{
$value\leftarrow \frac{LB_{i}^{2}}{c\times (\|o_{i1}\|_{1}+\|q\|_{1})^{2}}$;\\
$point\leftarrow o_{i1}$;\\
}
}
\Return $point$;
\end{algorithm}
Combining Quick-Probe and the aforementioned Condition A and Condition B, the searching process is described in Algorithm~\ref{QPMIPSearch}. Quick-Probe is applied to find the point $o$ as the input to determine the searching range in the projected space. During the range search in the projected space, when a point is returned, its inner product to the query point is recorded for the final verification. Besides, Condition A is also tested to determine whether to terminate the searching process (Unlike Condition B, Condition A doesn't require too much computation).

Because the searching range obtained by Quick-Probe is an estimated value, the obtained point may not satisfy Condition B completely, which indicates that the searching range may not completely guarantee $c$-AMIP point with the given probability $p$. Faced with this problem, we compensate it by expanding the searching range to ensure the probability-guaranteed $c$-AMIP result. If the entire range search has been performed, the recorded maximum inner product is brought into Condition B to test whether the result satisfies the condition. If satisfied, terminate the searching process and return the result. If not satisfied, according to Formula~\ref{Condition2}, the searching range will be extended to $r^{'}=\sqrt{\Psi_{m}^{-1}(p)\times({\|o_{M}\|}^{2}+{\|q\|}^{2}-\frac{2{\langle o_{max},q\rangle}}{c})}$ as compensation to find the final results. Since the obtained maximum inner product later is greater than or equal to current $\langle o_{max},q\rangle$, the extended $r^{'}$ is larger than or equal to the actual searching range satisfying the probability-guaranteed requirements.

\begin{algorithm}
\caption{\textsf{MIP-Search-\uppercase\expandafter{\romannumeral2}} ($D,c,p,q,o$)}\label{QPMIPSearch}
\LinesNumbered
$r\leftarrow dis(P(o),P(q))$; // Determined searching range\\
$o_{max}\leftarrow Null$;\\
$i\leftarrow 0$;\\
// Perform range search\\
\While{$dis(P(o_{i}),P(q))<r$}{
$i\leftarrow i+1$;\\
$o_{i}\leftarrow$ the original form of $P(o_{i})$;\\
\If{$\langle o_{i},q\rangle>\langle o_{max},q\rangle$}{
$o_{max}\leftarrow o_{i}$; // Update MIP point\\
\If{Condition A}{
\Return $o_{max}$;
}
}
}
\If{Condition B}{
\Return $o_{max}$;
}
\Else{
Update the searching range to $r^{'}$;\\
\While{$dis(P(o_{i}),P(q))<r^{'}$}{
$i\leftarrow i+1$;\\
$o_{i}\leftarrow$ the original form of $P(o_{i})$;\\
\If{$\langle o_{i},q\rangle>\langle o_{max},q\rangle$}{
$o_{max}\leftarrow o_{i}$;\\
\If{Condition A}{
\Return $o_{max}$;
}
}
}
}
\Return $o_{max}$;
\end{algorithm}

\subsection{Optimized Projected Dimension}\label{subsec:OPD}
In Quick-Probe, the projected points are transformed into binary codes. It indicates that $m$ projected dimensions will bring $2^{m}$ binary codes. If assuming that each binary code represents the same number of points, $2^{m}$ groups will bring $n/2^{m}$ points in each group. It can be observed that more projected dimensions may bring more groups, while bring fewer points in each group. If the point can be located by directly searching one group, fewer points in one group will lead to less time consumption. However, more groups also require more time to compute their lower bounds. Therefore, there exists a trade-off and we can derive an optimized projected dimension to improve the efficiency of Quick-Probe.

Binary codes with $m$ bits can divide the whole dataset into up to $2^{m}$ groups. The time consumption to compute the groups' lower bounds and find the group with the smallest lower bound is $2^{m}(m+1)$. We assume that the whole dataset can be equally divided and the point satisfying Formula~\ref{Condition2a} can be located by searching only one group. Therefore, each group contains $n/2^{m}$ points and the time consumption of searching the point is $n/2^{m}$. The total time consumption is $2^{m}(m+1)+n/2^{m}$. We set the function $f(m)=2^{m}(m+1)+n/2^{m}$. Since the second derivative of $f(m)$ is greater than 0, $f(m)$ has the minimum value. Our objective is to compute $m=\arg \min f(m)$, which is considered as the optimized projected dimension.

\section{Index Structure}\label{sec:index}
In the standard iDistance shown in Fig.~\ref{IDistance}, when performing range search, the searching area is much larger than the given searching sphere, which indicates that a large portion of searching area is unnecessary.

Different from the standard iDistance, to avoid much unnecessary searching area, we adopt a different partition pattern as shown in Fig.~\ref{ParIDistance}. We use the following Formula~\ref{equation7} to compute each point's index key,
\begin{equation}\label{equation7}
\textit{I}(p)=\lfloor i*C+dis(p,O_{i})/\varepsilon\rfloor
\end{equation}
where $\varepsilon$ is a constant determined by the data distribution. In detail, taking a two-dimensional space as an example, we obtain the clusters' radii after the first stage of clustering and compute their average. Then, we make a circle with the average as the radius denoted as $r_{avg}$, and the value of $\varepsilon$ is equal to $r_{avg}/N_{key}$ to divide the circle into $N_{key}$ rings with equal ring widths, which also means that points can be mapped to $N_{key}$ keys. We continue to employ $k$-means to divide the sets of points in the rings into several sub-partitions, while the clusters' centers and radii are the sub-partitions' pivots and radii, respectively. In the searching process, points can be filtered in sub-partitions by whether they intersect with the given searching sphere. In addition, the points in the same sub-partition can be collectively organized on disks in order, which means that the adjacent points belonging to the same sub-partition are likely to be organized on the same disk, while the adjacent sub-partitions are also likely to be organized on the adjacent disks. It's beneficial to reduce page accesses since points can be read from disks in sub-partitions to avoid random readings. As shown in Fig.~\ref{ParIDistance}, an index key indexes a deep grey ring in the partition. The points in this ring are divided into eight sub-partitions. Given a searching sphere centered at the query point, two of the eight sub-partitions intersect with the given sphere and the points in these sub-partitions are selected as the candidate points.
\begin{figure}\centering
\includegraphics[width=0.45\textwidth]{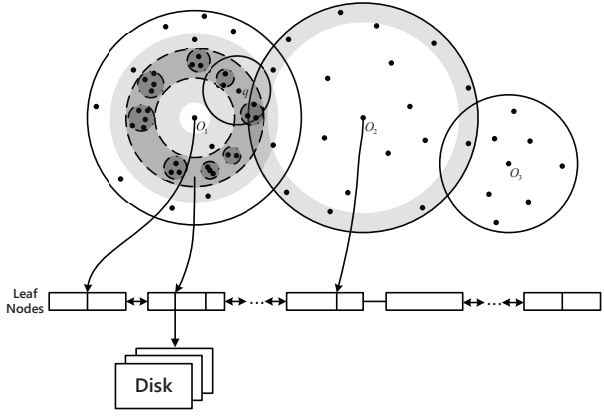}
\caption{IDistance with New Partition Pattern} \label{ParIDistance}
\vspace{-12pt}
\end{figure}
In our partition pattern, it's required to select appropriate values of the number of partitions $k_{p}$ and the number of sub-partitions $k_{sp}$ to ensure that each sub-partition contains a certain number of points to make the filter effective. To the end, we introduce a parameter called selectivity $\mu$. That is, we try to make nearly $\mu n$ points in each sub-partition by setting appropriate $k_{p}$ and $k_{sp}$. We assume that, after the first clustering stage via $k_{p}$-means, the number of points in each cluster is the same, which is $\frac{n}{k_{p}}$. We determine the value of $\varepsilon$ in Formula~\ref{equation7} according to the data distribution to control the number of keys in each cluster, and the number of points corresponding to each key is also assumed to be the same. We denote the number of keys in a cluster as $N_{key}$, thereby the number of points represented by a key is $\frac{n}{k_{p}*N_{key}}$. Based on the aforementioned assumptions, the number of points in each sub-partition is $\frac{n}{k_{p}*N_{key}*k_{sp}}$ after clustering by $k_{sp}$-means. Therefore, the selectivity $\mu=\frac{1}{k_{p}*N_{key}*k_{sp}}$. In the experimental evaluations, we will give the parameter settings on the testing datasets.

\color{black}Algorithm~\ref{alg:paridistance} introduces the index construction containing the dividing process, computing the index keys and constructing the B+-tree to index these points.
\begin{algorithm}
\caption{\textsf{Index-Construct}($D$)}\label{alg:paridistance}
\LinesNumbered
Project original dataset $D$ onto projected dataset $D_{p}$;\\
Divide $D_{p}$ into $k_{p}$ partitions $\{P_{1},P_{2},...,P_{k_{p}}\}$;\\
\For{$i=1$ to $k_{p}$}{
\For{every point $p$ in $P_{i}$}{
$\textit{I}(p)=\lfloor i*C+dis(p,O_{i})/\varepsilon\rfloor$; // Formula~\ref{equation7}\\
}
Divide points with the same index keys in $P_{i}$ into $k_{sp}$ sub-partitions;\\
}
Construct B+-tree index and organize points on disks;\\
\end{algorithm}
\vspace{-20pt}
\section{Time and Space Complexities}\label{sec:complexities}
The time cost of our method consists of five parts. Firstly, according to Section~\ref{sec:quickprobe}, we have the computed optimized projected dimension $m=O(\log n)$ and the time cost of locating the point through Quick-Probe is $2^{m}m+2^{m}+\frac{n}{2^{m}}=O(n\log n)$. Secondly, the time complexity of computing $q$'s projection is $O(d)$. Thirdly, the time cost of locating the partition containing the projected query point and computing the projected query point's key is $k_{p}m+1=O(1)$. Then, since there are $k_{p}N_{key}$ keys in B+-tree, locating the key in the B+-tree costs $\log (k_{p}N_{key})$. In the B+-tree, assuming that $\alpha k_{p}N_{key}$ ($0<\alpha<1$) keys are searched, it costs $\alpha k_{p}N_{key}\log (k_{p}N_{key})$. The process of determining whether the searching range intersects with $\alpha k_{p}N_{key}k_{sp}$ sub-partitions costs $\alpha k_{p}N_{key}k_{sp}m$. Summing them up, the whole searching process costs $\log k_{p}N_{key}+\alpha k_{p}N_{key}\log k_{p}N_{key}+\alpha k_{p}N_{key}k_{sp}m=O(\log n)$. Finally, we denote that the filtering rate is $\beta$ ($0<\beta<1$), which indicates $\beta n$ are selected as candidate points. Hence computing the inner products for candidate points costs $\beta nd=O(d)$. Therefore, the time complexity of our method is $O(n\log n+d+1+\log n+d)=O(d+n\log n)$.

We also analyze the space cost of our proposed method. The space complexity of our method consists of the space complexities of storing $n$ original high-dimensional points and $n$ projected low-dimensional points, which are $O(nd)$ and $nm=O(n\log n)$, respectively. In addition, In Quick-Probe, the space complexity of storing the binary codes and each point $o$'s $\|o\|_{1}$ are $nm=O(n\log n)$ and $O(n)$, respectively. Thus, the total space cost is $O(nd+n\log n+n\log n+n)=O(nd+n\log n)$.

We also list the time and space complexities of two benchmark methods, H2-ALSH~\cite{DBLP:conf/kdd/HuangMFFT18} and Norm Ranging-LSH~\cite{DBLP:conf/nips/YanLDCC18} in Table~\ref{Complexities}. From Table~\ref{Complexities}, the time complexity of our method outperforms two benchmark methods. Although the space complexities of three methods are the same, in fact, the projected space in our method is much smaller than the number of hash tables in H2-ALSH or the hash codes' length in Norm Ranging-LSH.
\vspace{-10pt}
\begin{table}[H]
\centering
\caption{Time and Space Complexities}
\label{Complexities}
\begin{tabular}{lll}
\hline\noalign{\smallskip}
 & Time Complexity & Space Complexity \\
\noalign{\smallskip}\hline\noalign{\smallskip}
ProMIPS & $O(d+n\log n)$ & $O(nd+n\log n)$ \\
L2-ALSH & $O(d\log n+n\log n)$ & $O(nd+n\log n)$ \\
Norm Ranging-LSH & $O(d\log n+n\log n)$ & $O(nd+n\log n)$ \\
\noalign{\smallskip}\hline
\end{tabular}
\vspace{-10pt}
\end{table}

\section{Experimental Evaluations}\label{sec:experimentalevaluations}
\vspace{-5pt}
\subsection{Experiment Setup}
\vspace{-5pt}
\subsubsection{Benchmark Methods}
We select two state-of-the-art methods with probability guarantee in accuracy, H2-ALSH~\cite{DBLP:conf/kdd/HuangMFFT18} and Norm Ranging-LSH~\cite{DBLP:conf/nips/YanLDCC18}, as two benchmark methods. In addition, to compare our method with the method without probability guarantee in accuracy, we adopt the asymmetric transformation in H2-ALSH to convert MIP search problem into NN search problem, and select the latest product quantization-based NN search technique~\cite{DBLP:conf/cvpr/KalantidisA14} which performs well in accuracy and efficiency to solve the problem as a benchmark method. In the experiments, our method is denoted as ``ProMIPS". Three benchmark methods are denoted as ``H2-ALSH", ``Range-LSH" and ``PQ-Based", respectively. To evaluate the page access, we employ the disk-resident QALSH in the implementation of H2-ALSH. In Range-LSH, we organize the data in each subset sequentially on disks according to the descending order of each subset's maximum norm. In PQ-based method, we organize the data according to each cell's inverted list. All methods are implemented in Java and all experiments are conducted on an ECS with Intel Core Processor (Haswell, no TSX) 2.29GHZ, 48GB main memory, and 512GB hard disk, running under Windows 10. We use the buffering management in the operating system.

\subsubsection{Datasets and queries}
Four real datasets Netflix~\cite{bennett2007netflix}, Yahoo~\cite{dror2011yahoo}, P53\footnote{http://archive.ics.uci.edu/ml/datasets/p53+Mutants} and Sift\footnote{http://archive.ics.uci.edu/ml/datasets/SIFT10M} are summarized in Table~\ref{Dataset}. On Netflix and Yahoo, the user vectors and item vectors are generated by PureSVD~\cite{DBLP:conf/kdd/HuangMFFT18, cremonesi2010performance}. For all datasets, 100 points are randomly selected as the query points.

\begin{table}[H]
\centering
\caption{Datasets}
\label{Dataset}
\begin{tabular}{llll}
\hline\noalign{\smallskip}
Parameter & $n$ & $d$ & Data Size \\
\noalign{\smallskip}\hline\noalign{\smallskip}
Netflix & 17770 & 300 & 84.2MB \\
Yahoo & 624961 & 300 & 2.3GB \\
P53 & 31420 & 5408 & 1.07GB\\
Sift & 11164866 & 128 & 7.3GB\\
\noalign{\smallskip}\hline
\end{tabular}
\vspace{-5pt}
\end{table}

\subsubsection{Evaluation metrics}
\begin{itemize}
  \item Index Size. It is defined as the size of each evaluated method's index.
  \item Pre-processing Time. It is defined as the pre-computation and the index construction time of each evaluated method.
  \item Overall Ratio. It is defined as: $\frac{1}{k}\sum_{i=1}^{k} \frac{\langle o_{i},q\rangle}{\langle o_{i}^{*},q\rangle}$ in $c$-$k$-AMIP search problem, where $o_{i}$ is the $i$-th returned AMIP point and $o_{i}^{*}$ is the exact $i$-th MIP point of the query point. Intuitively, the overall ratio is between 0 and 1 and a larger overall ratio indicates a higher accuracy.
  \item Recall. It is defined as: $t/k$ in $c$-$k$-AMIP search problem. $t$ is the number of the returned AMIP points which are actually in the set of exact $k$-MIP points. A larger recall means more exact $k$-MIP points are returned, indicating a higher accuracy.
  \item Page Access. It is defined as the number of disk pages to be accessed during the searching process.
  \item CPU Time. It is defined as the CPU time for performing a $c$-$k$-AMIP search.\
  \item Total Time. It is defined as the running time for reading data from disks and performing a $c$-$k$-AMIP search.
\end{itemize}
\begin{figure}
\hspace{-10pt}
\subfigure[Index Size]{
\includegraphics[width=.25\textwidth]{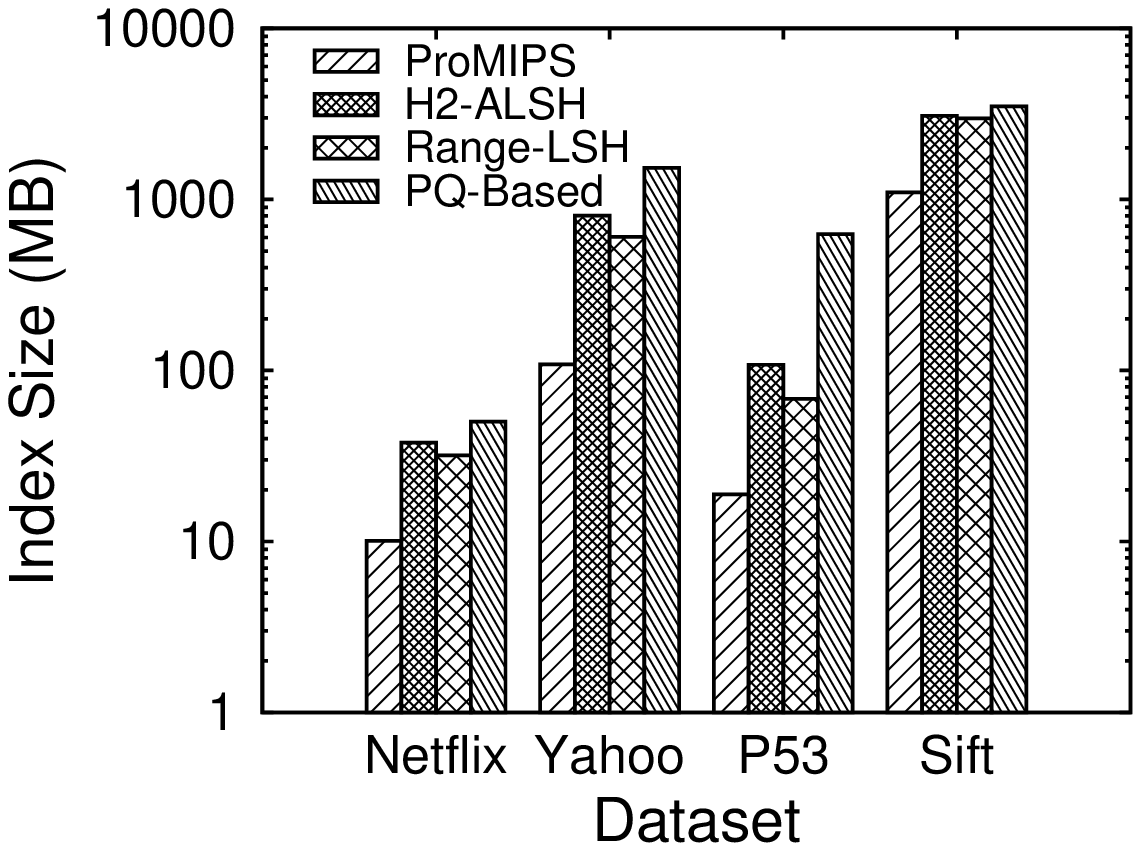}}
\hspace{-10pt}
\subfigure[Pre-processing Time]{
\includegraphics[width=.25\textwidth]{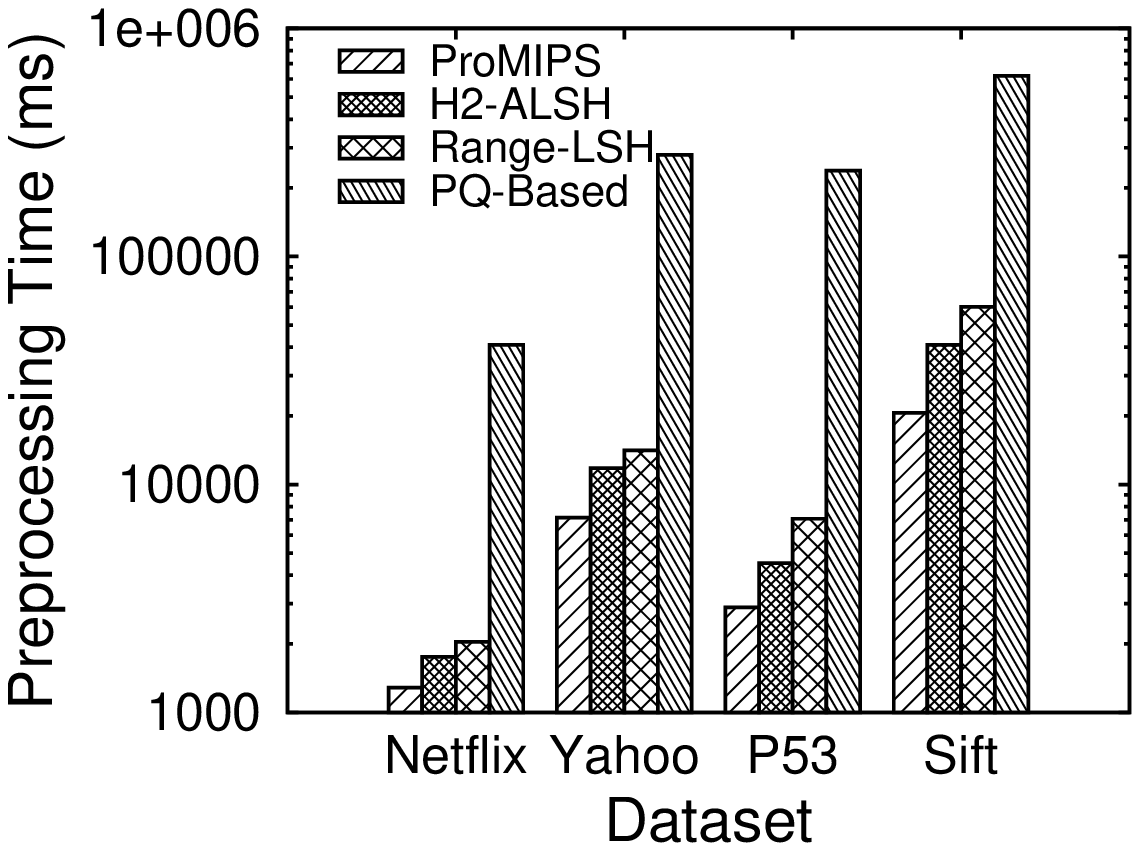}}
\caption{Index Size and Pre-processing Time}
\label{Preprocessing}
\vspace{-15pt}
\end{figure}
\subsubsection{Parameter Settings}
\begin{figure*}
\hspace{-10pt}
\subfigure[Netflix]{
\includegraphics[width=.25\textwidth]{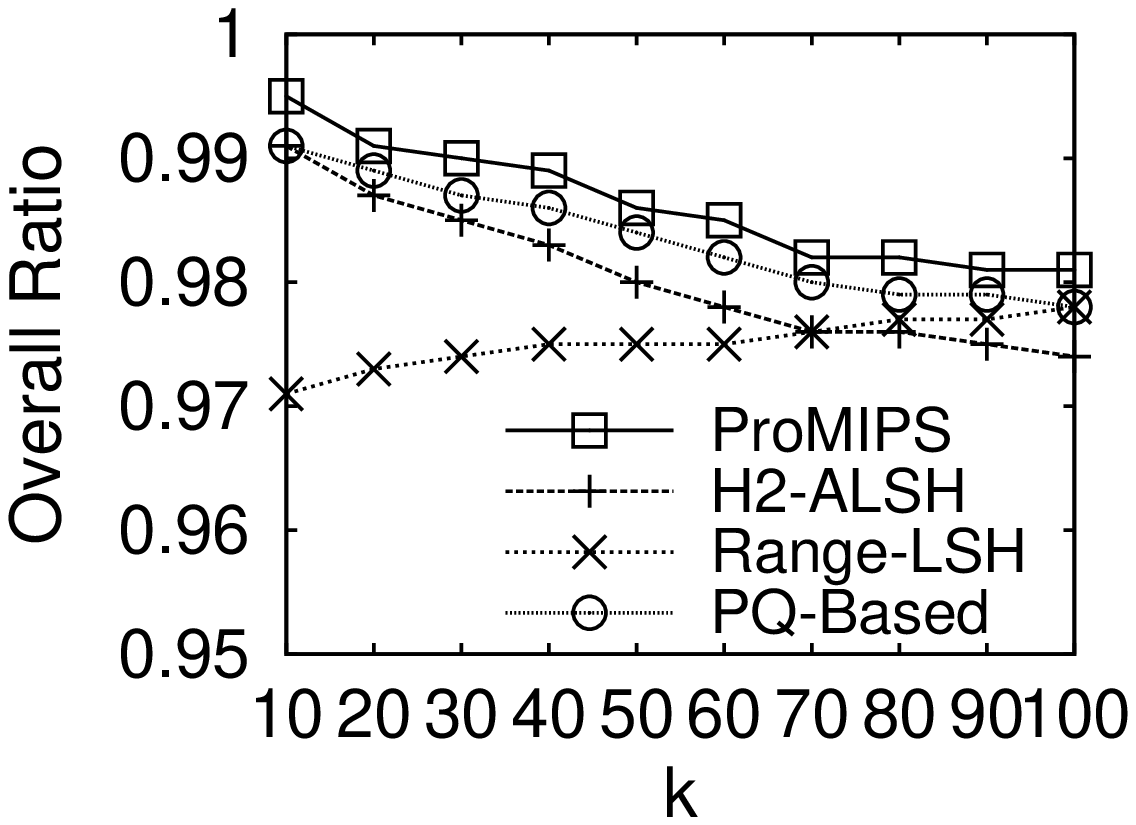}}
\hspace{-10pt}
\subfigure[Yahoo]{
\includegraphics[width=.25\textwidth]{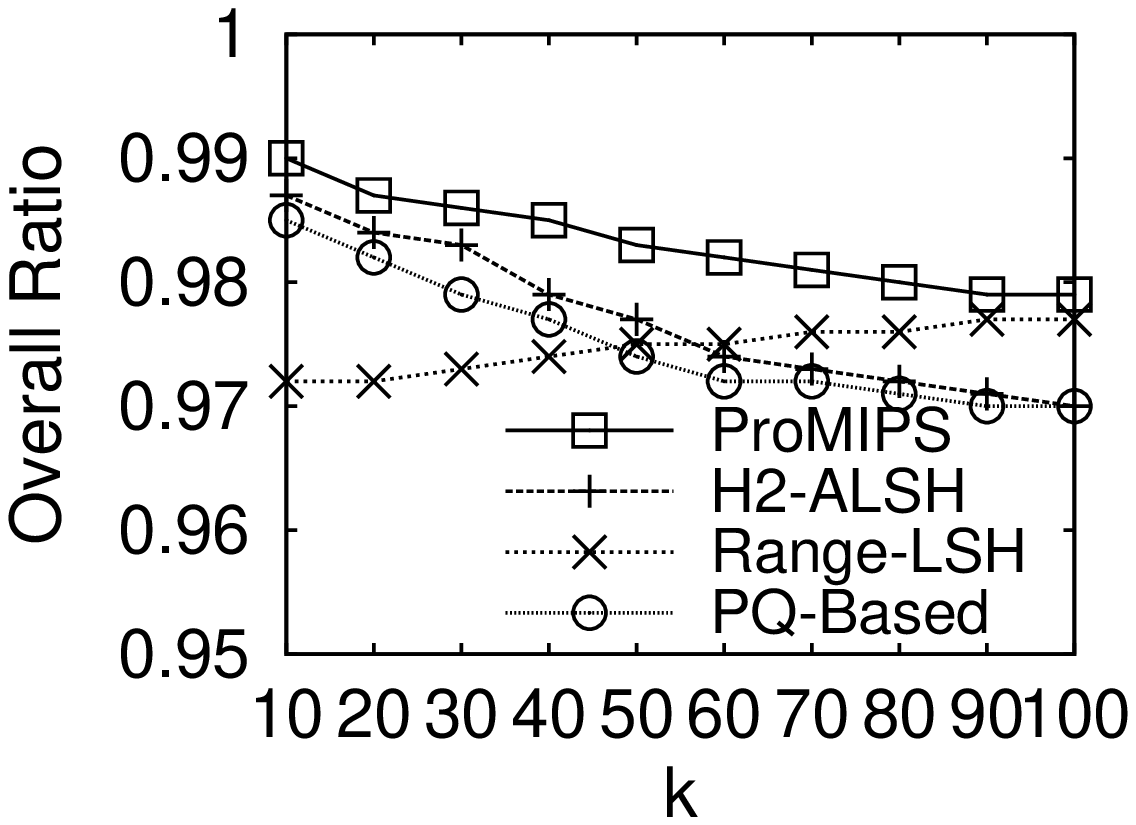}}
\hspace{-10pt}
\subfigure[P53]{
\includegraphics[width=.25\textwidth]{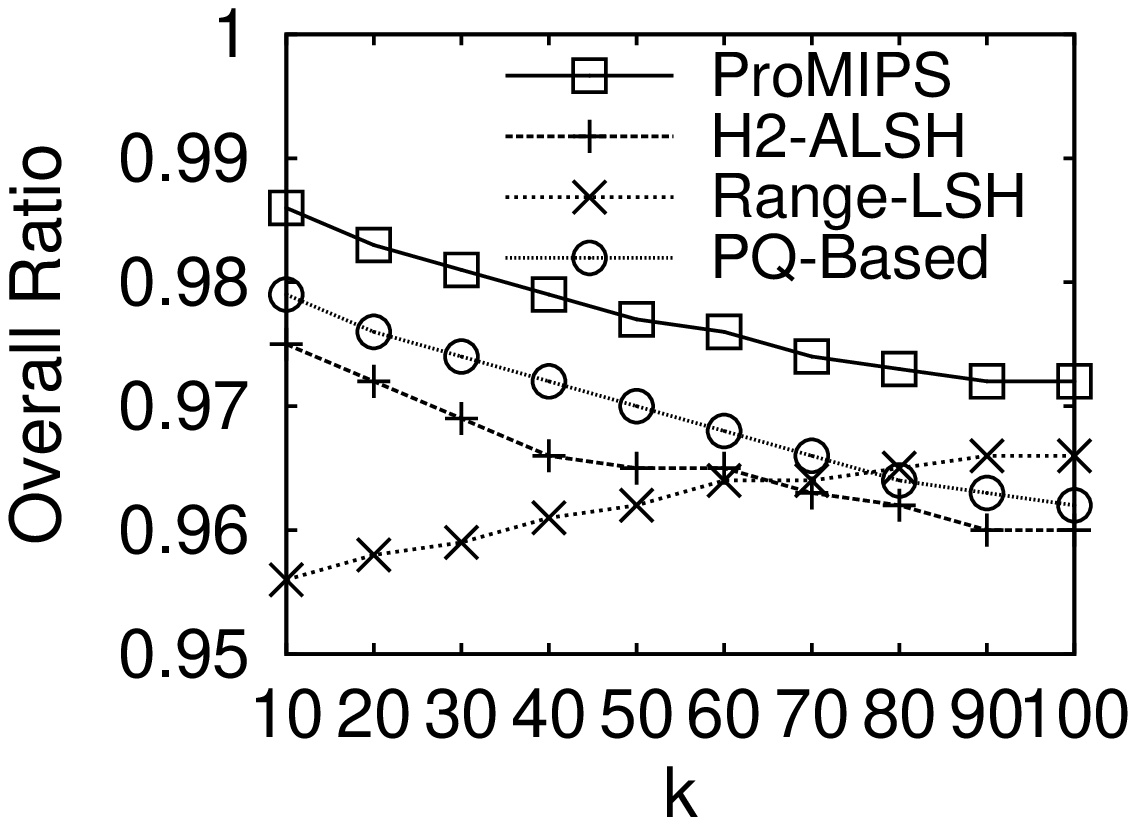}}
\hspace{-10pt}
\subfigure[Sift]{
\includegraphics[width=.25\textwidth]{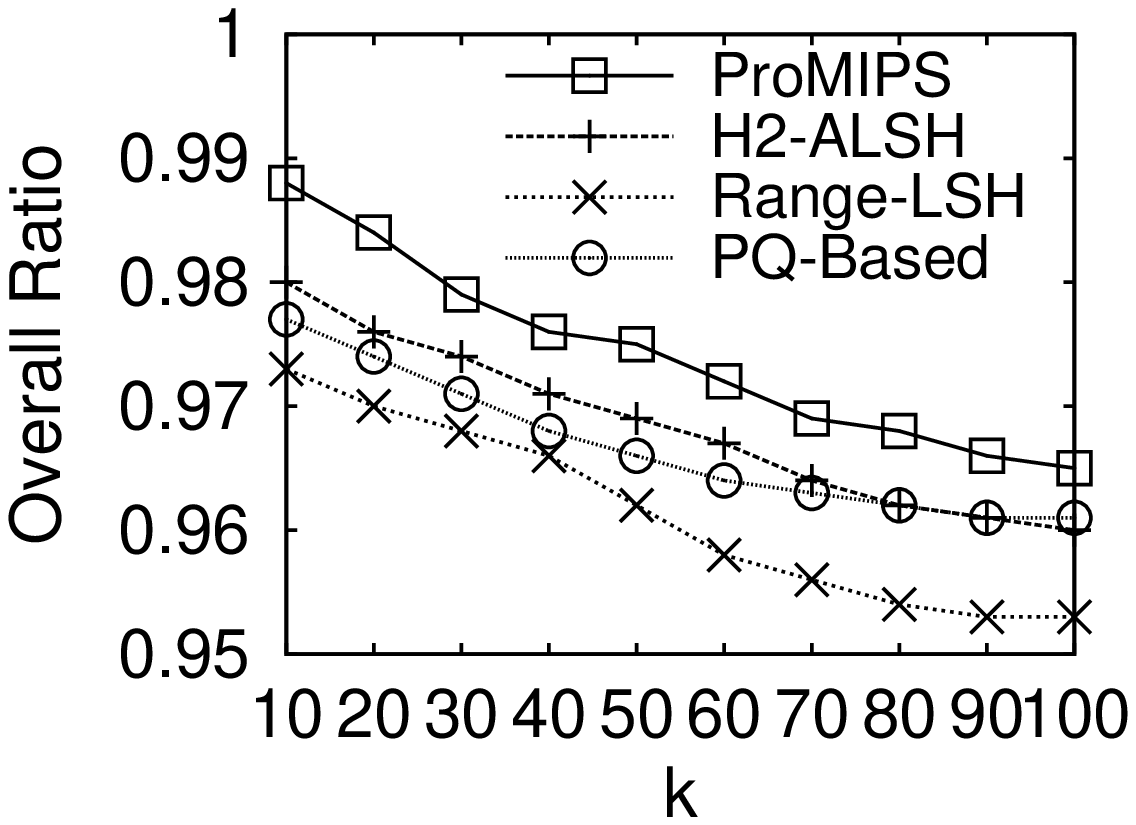}}
\caption{Overall Ratio}
\label{Overall_Ratio}
\vspace{-5pt}
\end{figure*}

\begin{figure*}
\hspace{-10pt}
\subfigure[Netflix]{
\includegraphics[width=.25\textwidth]{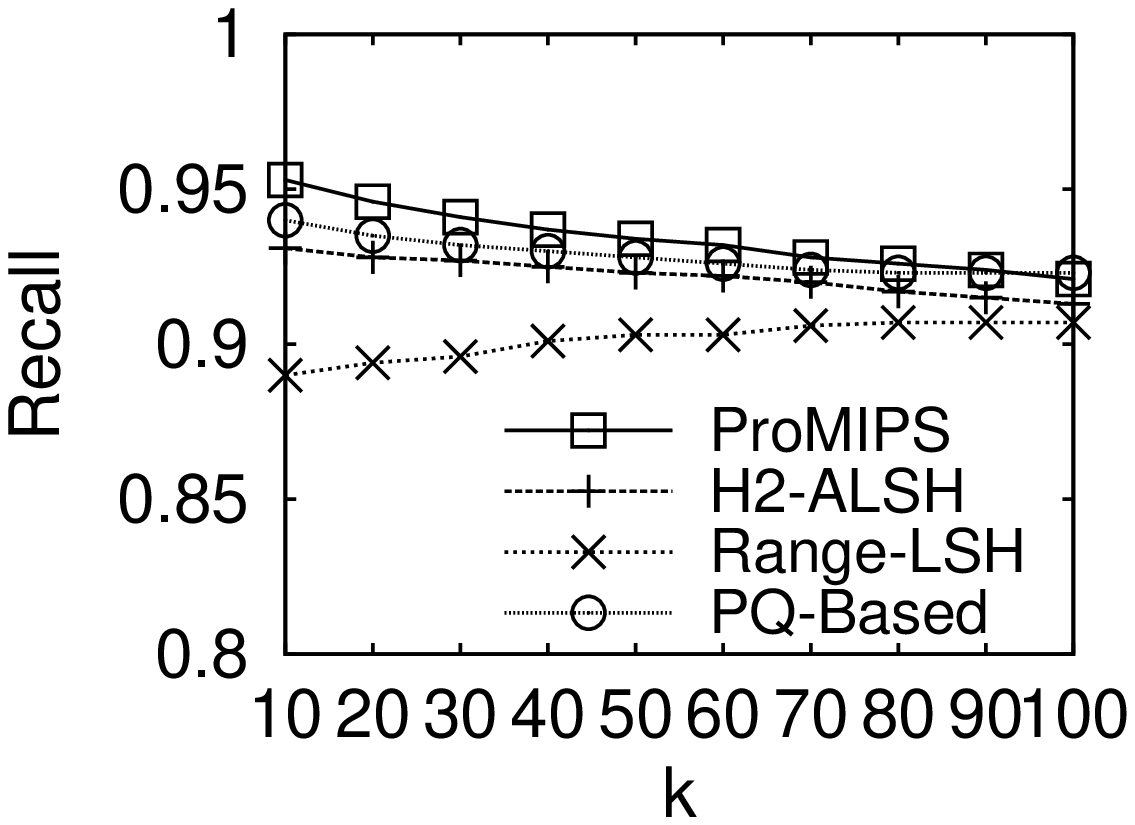}}
\hspace{-10pt}
\subfigure[Yahoo]{
\includegraphics[width=.25\textwidth]{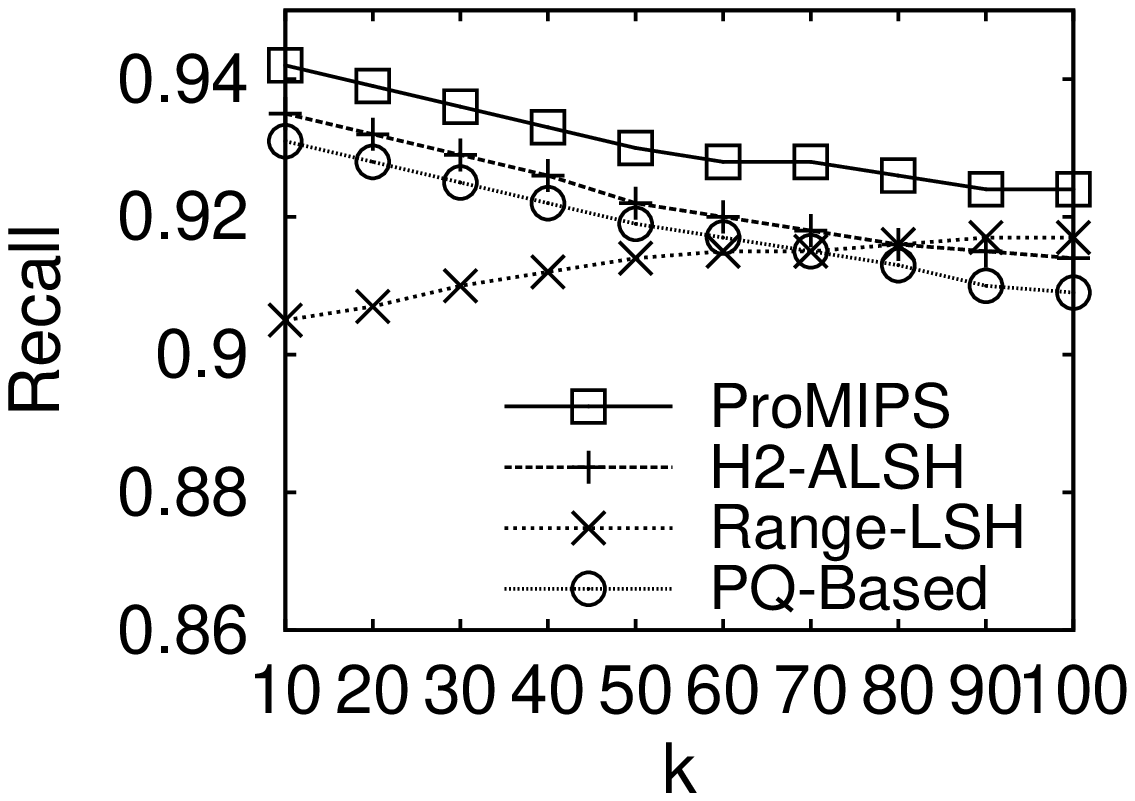}}
\hspace{-10pt}
\subfigure[P53]{
\includegraphics[width=.25\textwidth]{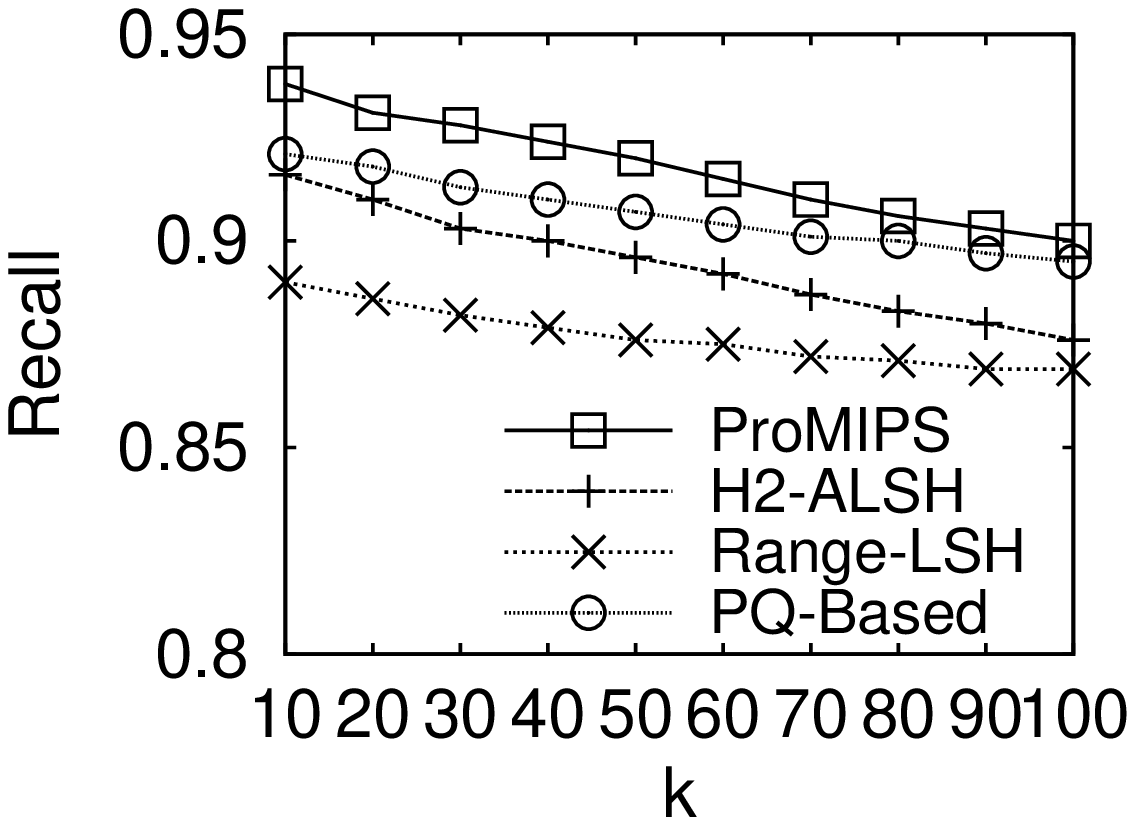}}
\hspace{-10pt}
\subfigure[Sift]{
\includegraphics[width=.25\textwidth]{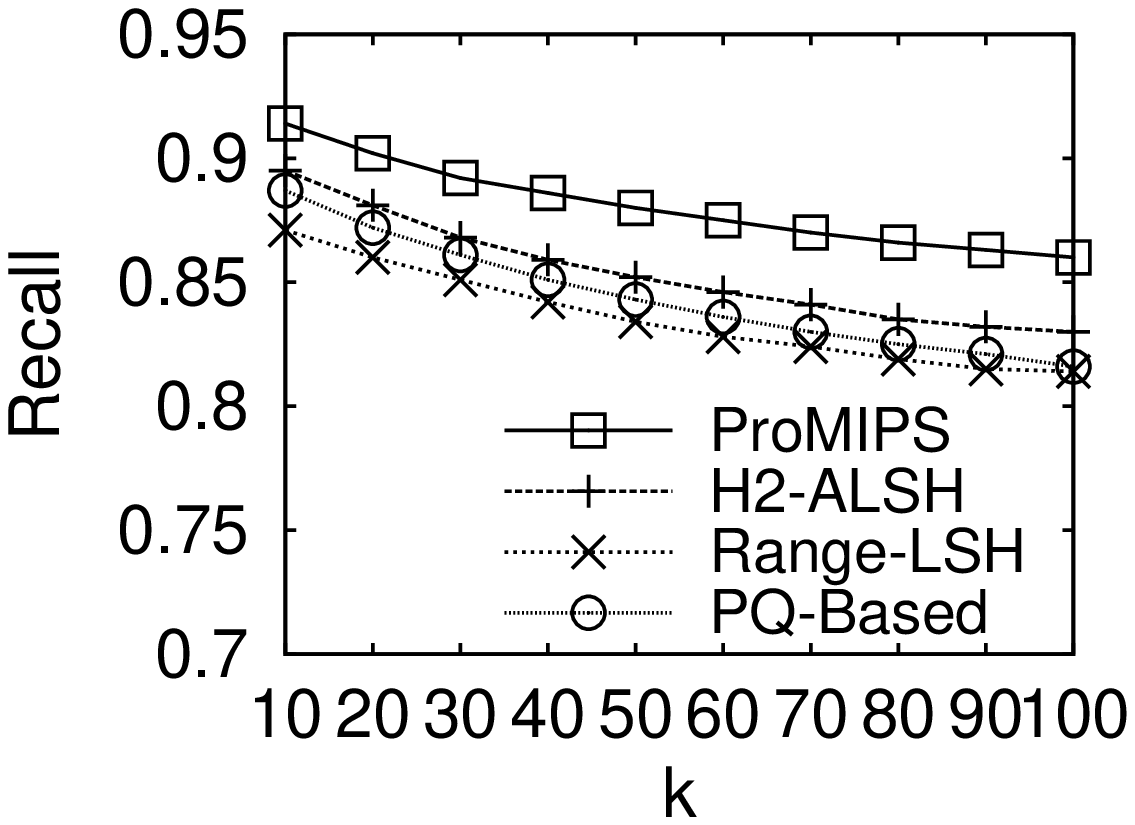}}
\caption{Recall}
\label{Recall}
\vspace{-12pt}
\end{figure*}
The performance of ProMIPS is evaluated under different parameter settings. According to Section~\ref{subsec:OPD}, the projected dimensions $m$ are set to 6 on Netflix and P53. On Yahoo and Sift, the projected dimensions $m$ are set to 8 and 10, respectively. Through experiments, we find that it doesn't have much effect on efficiency when the values of $k_{p}$, $N_{key}$ and $k_{sp}$ are set in the ranges of 5-15, 20-50 and 5-25, respectively. Therefore, we set $k_{p}=5$, $N_{key}=40$ and $k_{sp}=10$ as the default values for all testing datasets. The values of $\varepsilon$ on Netflix, Yahoo, P53 and Sift are 0.02, 40, 0.1 and 250, respectively. The default approximation ratio $c$ is set to 0.9 and we vary $c$ to 0.7, 0.8 and 0.9 to evaluate its impact on ProMIPS's searching accuracy and efficiency. The default guaranteed probability $p$ is set to 0.5 and we vary $p$ to 0.3, 0.5, 0.7 and 0.9 to evaluate its impact on ProMIPS's searching accuracy and efficiency. In H2-ALSH, the value of $c_{0}$ is fixed to 2.0~\cite{DBLP:conf/kdd/HuangMFFT18}. In Range-LSH, we divide the datasets into 32 partitions under a code length of 16~\cite{DBLP:conf/nips/YanLDCC18}. In PQ-based method, the whole space is divided into 16 subspaces. The number of centroids in each subspace is 256 and the number of searched nearest cells is 16 in the searching process~\cite{DBLP:conf/cvpr/KalantidisA14}. The required $k$ is set from 10 to 100 in all testing cases. When evaluating the page access, the disk page's size is set to 4KB on Netflix, Yahoo and Sift. On P53, the disk page's size is set to 64KB due to its high dimension.

\color{black}\subsection{Pre-processing Time and Index Size}
The pre-process of our method contains generating each point's projection, computing each point's norms and converting the projected points into binary codes for Quick-Probe, and constructing the index. The pre-processes of H2-ALSH and Range-LSH contain constructing multiple hash tables and transforming data points. The pre-process of PQ-based method contains constructing quantizers with multiple cells, computing the residuals, training for the rotation matrices and maintaining each cell's corresponding inverted list. The index size and the pre-processing time of four evaluated methods are illustrated in Figs.~\ref{Preprocessing}(a) and (b), respectively. On all datasets, the index size and the pre-processing time of ProMIPS beat other methods. This is because H2-ALSH and Range-LSH construct multiple hash tables and PQ-based method stores many local rotation matrices and cells incurring large space overheads, while ProMIPS constructs iDistance with a single B+-tree. In ProMIPS, although the two-stage dividing process in the index construction is time-consuming, only one B+-tree is required, which reduces the time overhead. Compared to H2-ALSH, Range-LSH uses more hash vectors to generate each point's bit vector and their proposed single-table multi-probe strategy requires more time to rank the hash tables. Therefore, it takes more pre-processing time in Range-LSH. Nevertheless, since the points' bit vectors take up less space, the index size of Range-LSH is smaller than that of H2-ALSH. Since the training process to obtain the optimized rotation matrices is costly and it's space-consuming to store rotation matrices and cells, the performances of PQ-based method on the index size and the pre-processing time are the worst.

\subsection{Overall Ratio and Recall}
Fig.~\ref{Overall_Ratio} reports the results on overall ratio when varying the value of $k$ from 10 to 100. Four methods perform well on all datasets while the values of overall ratio are over 0.95. From the experimental results, the overall ratio of ProMIPS is higher than those of the other three methods by up to $3\%$. Meanwhile, the overall ratio of ProMIPS is larger than the default approximation ratio when varying $k$. The phenomenon demonstrates that ProMIPS can guarantee $c$-$k$-AMIP search in accuracy. In addition, we test the recall of four methods on four datasets and the results are shown in Fig.~\ref{Recall}. In Fig.~\ref{Recall}, the similar trends are observed. Both of the experimental results on the overall ratio and recall illustrate that ProMIPS can provide $c$-$k$-AMIP point with a high accuracy.

\begin{figure*}
\hspace{-10pt}
\subfigure[Netflix]{
\includegraphics[width=.25\textwidth]{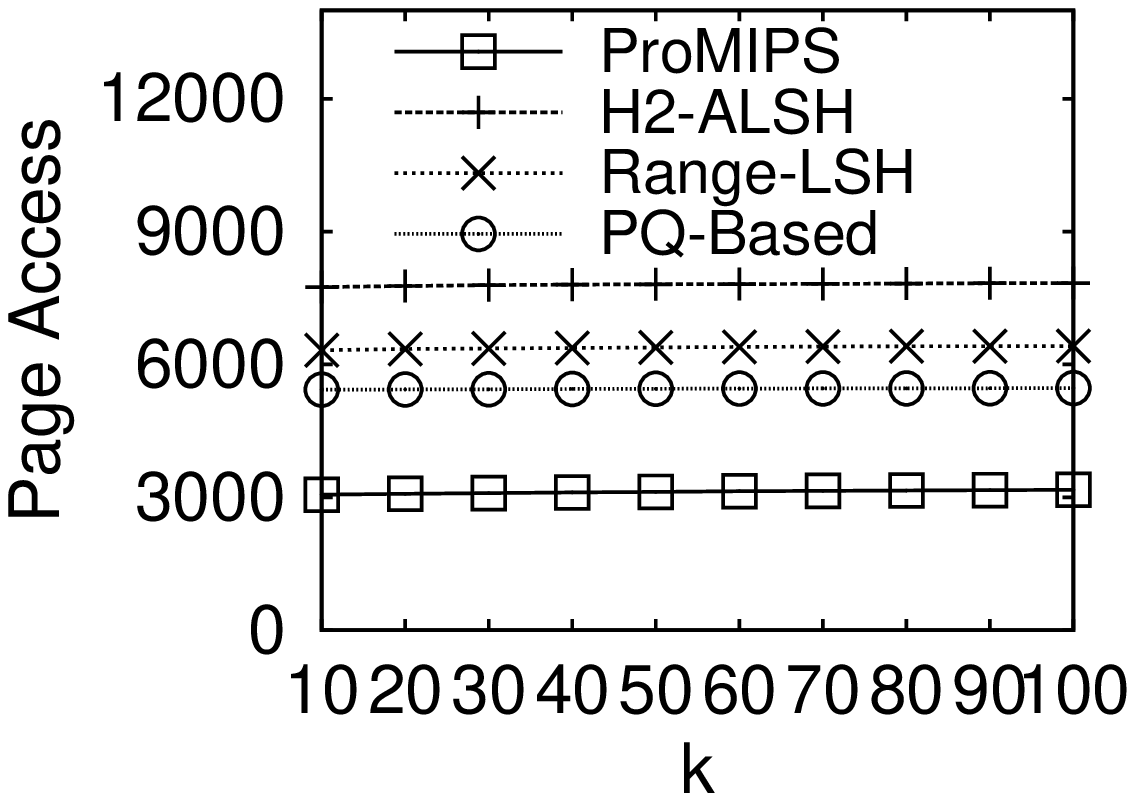}}
\hspace{-10pt}
\subfigure[Yahoo]{
\includegraphics[width=.25\textwidth]{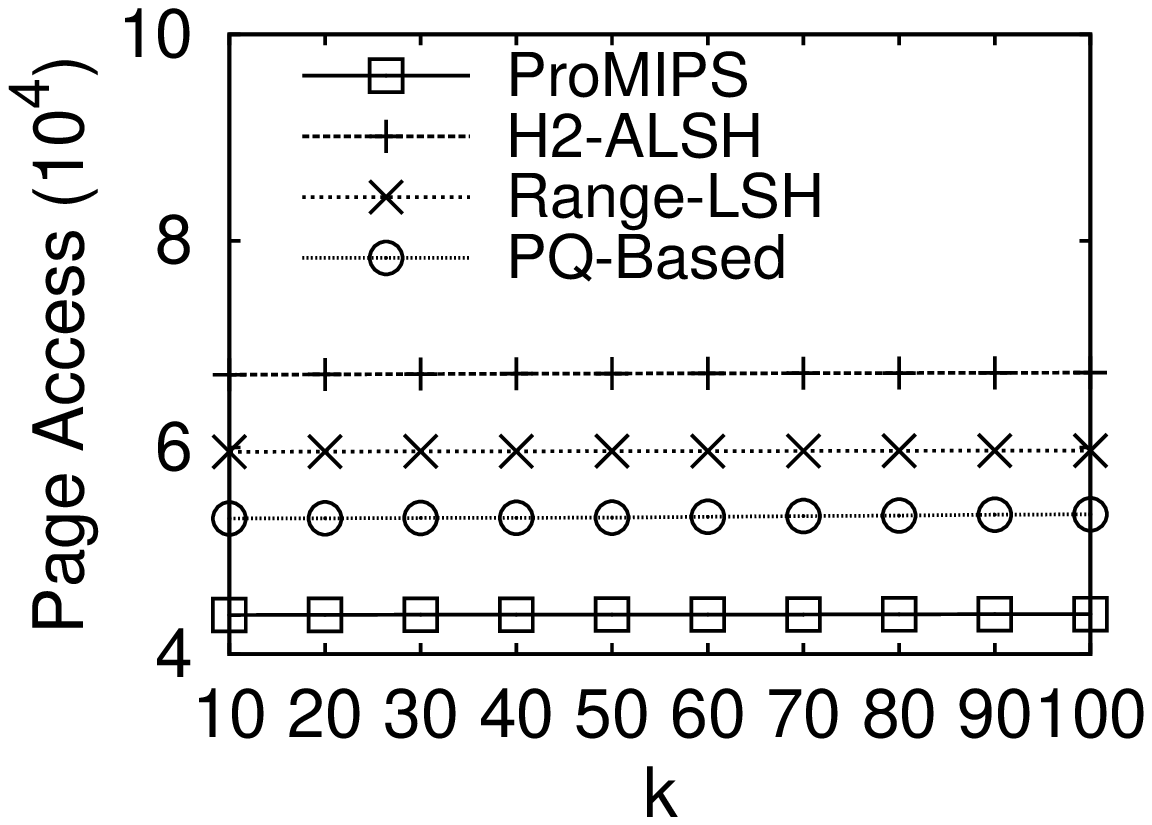}}
\hspace{-10pt}
\subfigure[P53]{
\includegraphics[width=.25\textwidth]{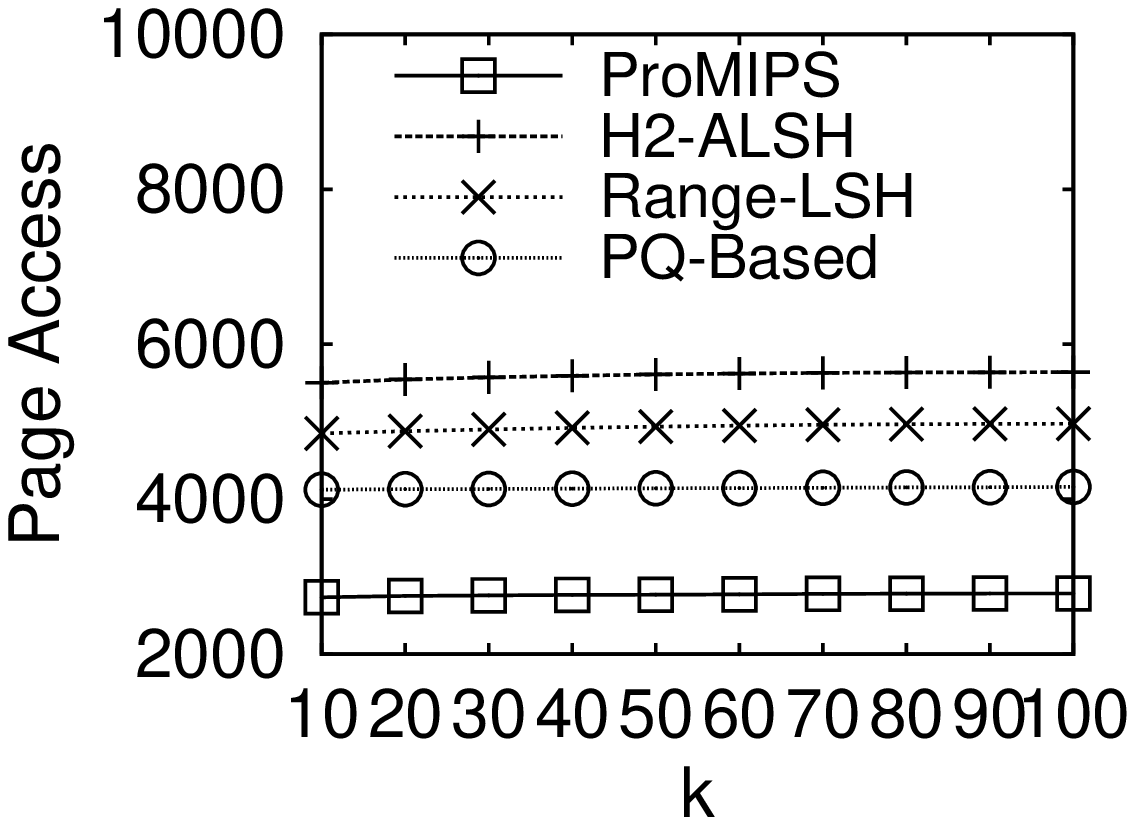}}
\hspace{-10pt}
\subfigure[Sift]{
\includegraphics[width=.25\textwidth]{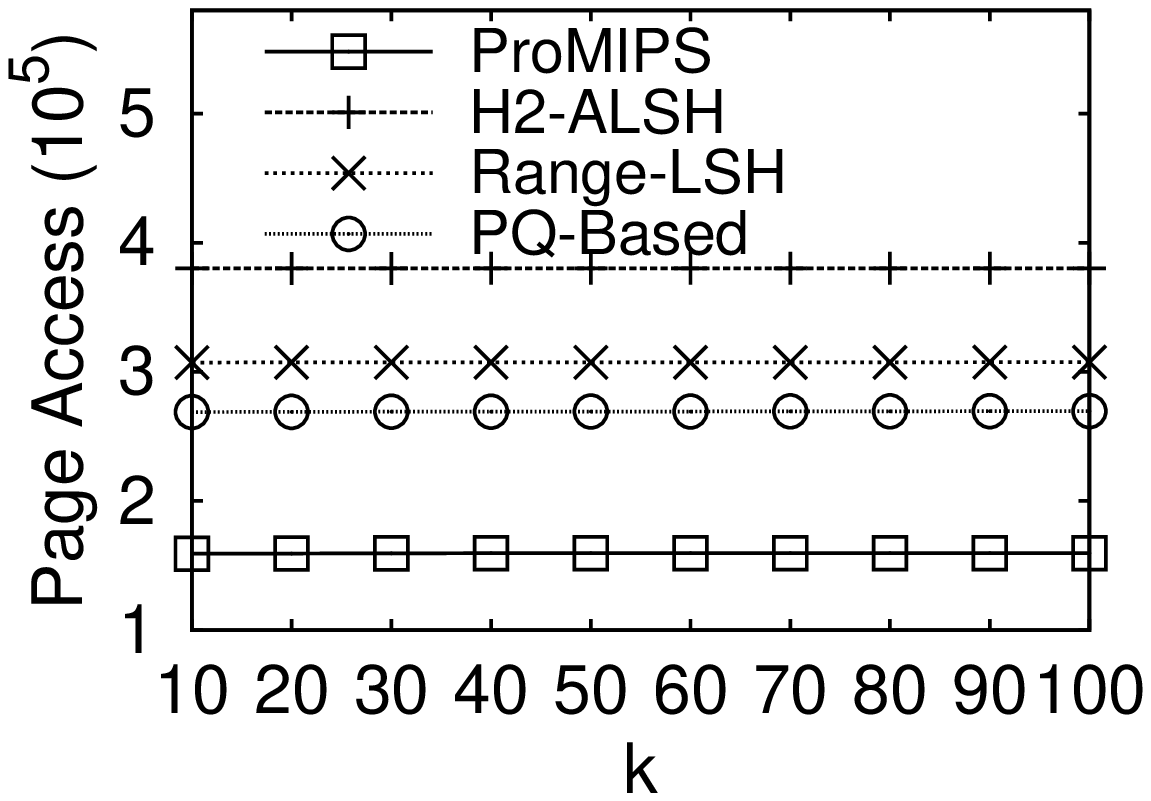}}
\caption{Page Access}
\label{IO_Cost}
\vspace{-5pt}
\end{figure*}

\begin{figure*}
\hspace{-10pt}
\subfigure[Netflix]{
\includegraphics[width=.25\textwidth]{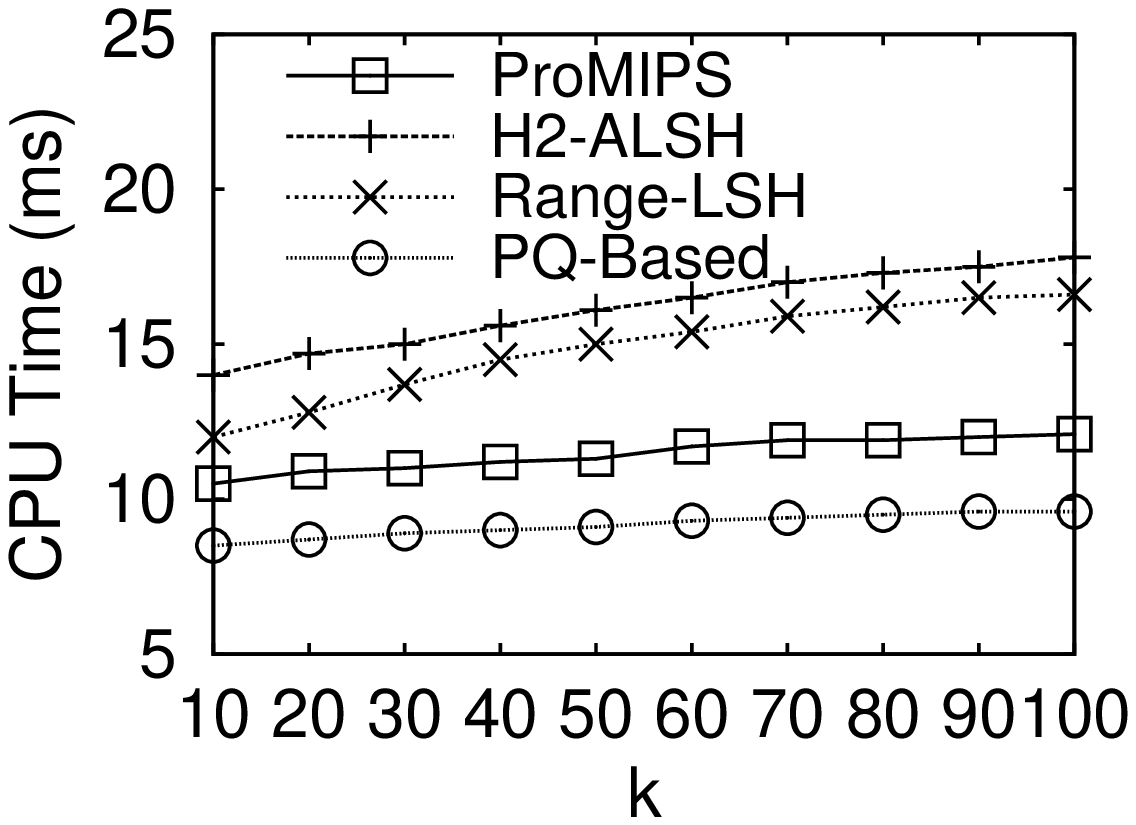}}
\hspace{-10pt}
\subfigure[Yahoo]{
\includegraphics[width=.25\textwidth]{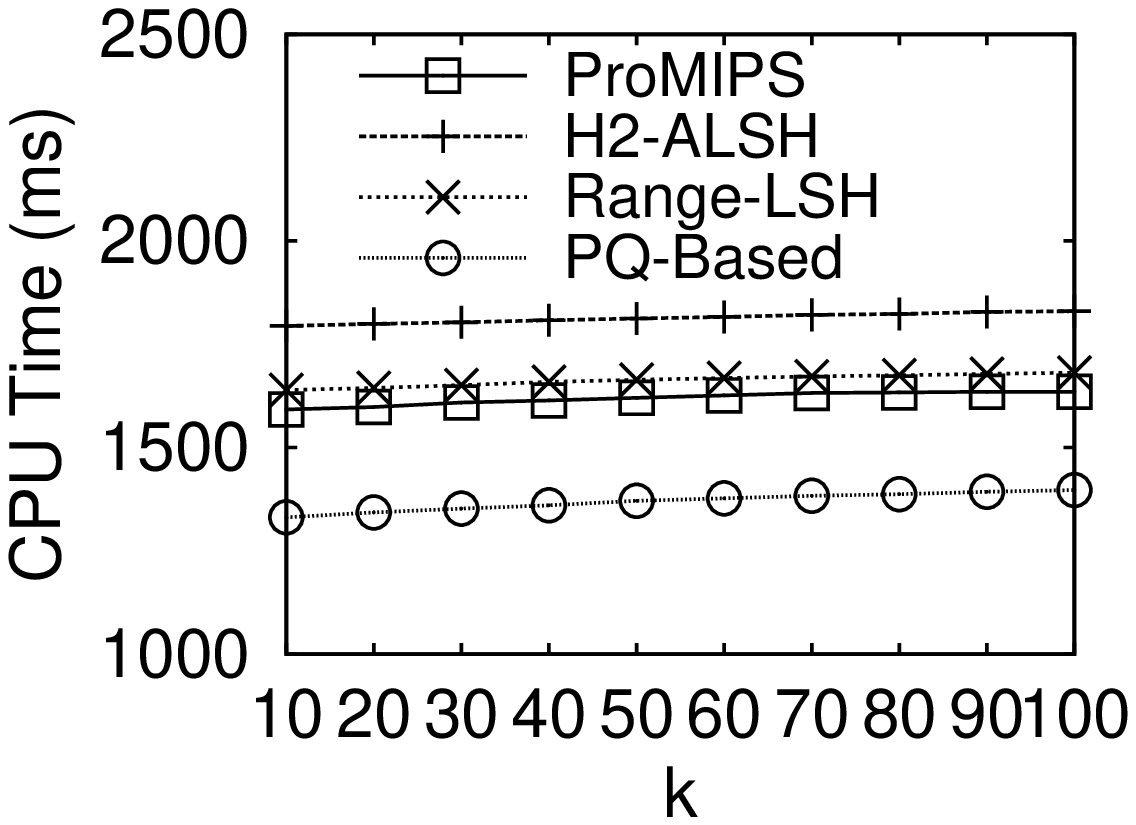}}
\hspace{-10pt}
\subfigure[P53]{
\includegraphics[width=.25\textwidth]{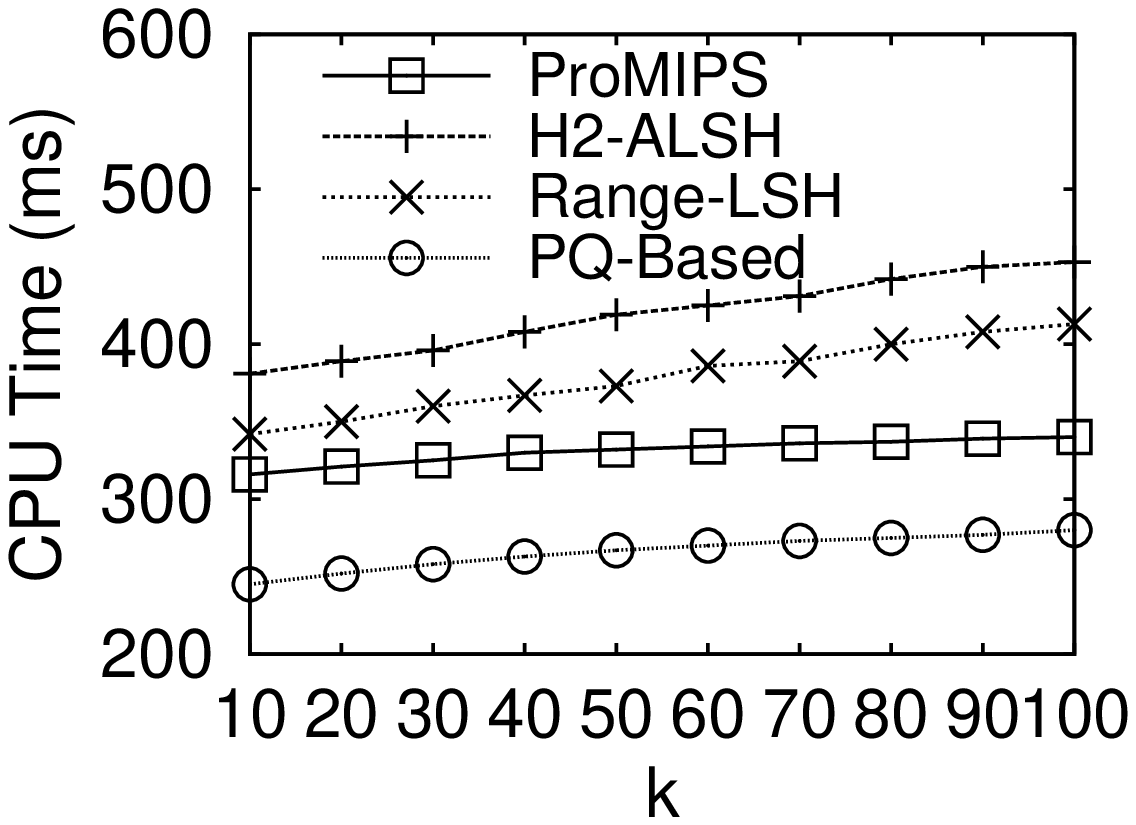}}
\hspace{-10pt}
\subfigure[Sift]{
\includegraphics[width=.25\textwidth]{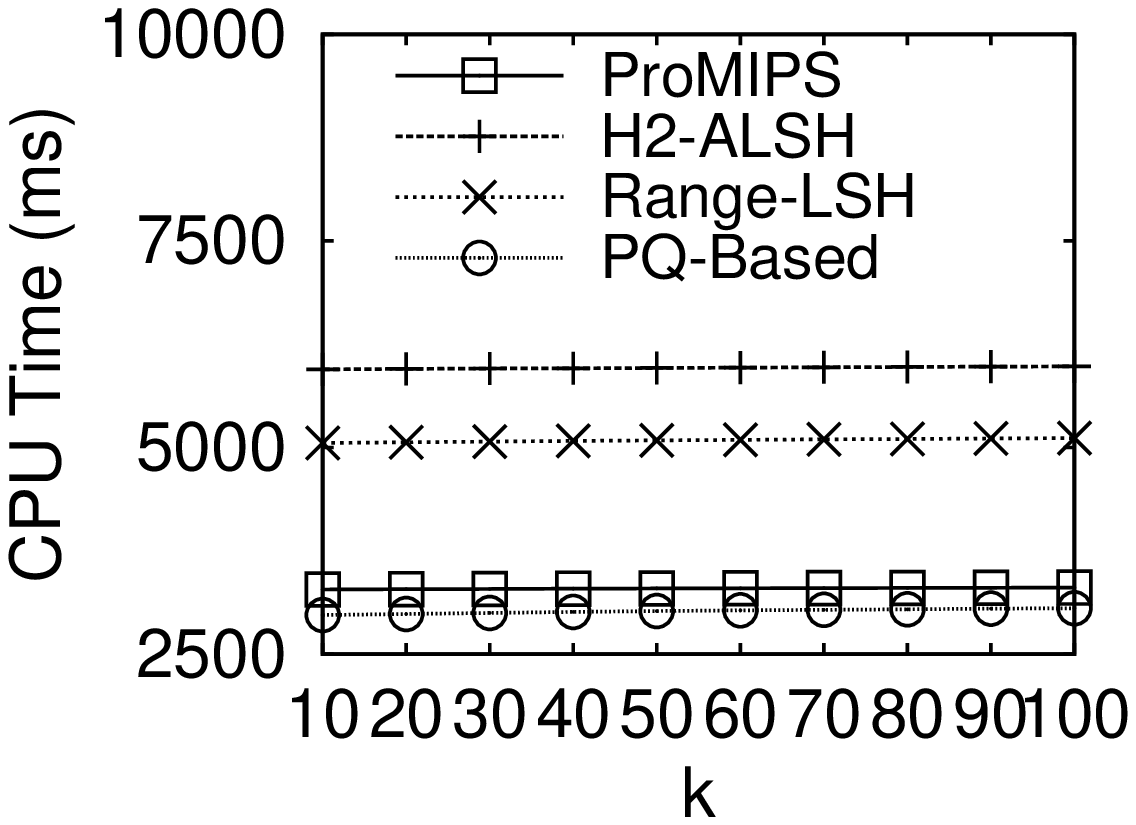}}
\caption{CPU's Time}
\label{Running_Time}
\vspace{-12pt}
\end{figure*}

\subsection{Page Access}
We evaluate the page access of four methods by varying $k$ from 10 to 100 as well and show the experimental results in Fig.~\ref{IO_Cost}. In Fig.~\ref{IO_Cost}, ProMIPS outperforms the other three methods in all testing cases as $k$ increases. It is because iDistance used in ProMIPS only requires one B+-tree as index, while both H2-ALSH and Range-LSH require more hash tables to ensure the accuracy, leading to more candidate points. In addition, the searching conditions in our method enable us to verify fewer candidate points to obtain satisfactory results. Meanwhile, benefiting from Quick-Probe, we can avoid reading the projected points from disks and testing them one by one. Besides, using the iDistance with our proposed partition pattern, the points can be collectively organized on disks in sub-partitions. The points can be read from disks sequentially to reduce page accesses. The experimental results on four datasets also illustrate that ProMIPS provides good efficiency in all data dimensions and at all data scales, which reflects our method's high scalability. In PQ-based method, we have to check many PQ-encoded residuals, which incurs more page accesses. Compared to H2-ALSH, Range-LSH performs better in terms of the page access because fewer hash buckets are probed during the searching process in Range-LSH benefitting from their proposed single-table multi-probe strategy, which brings fewer selected candidate points.

\subsection{CPU Time and Total Time}
In Fig.~\ref{Running_Time}, we evaluate the CPU time to test the efficiency of four methods. From the experimental results, the performance of ProMIPS on CPU time is comparable. PQ-based method performs the best on CPU time because the distances between PQ-encoded residuals are pre-computed in the pre-process. Compared to H2-ALSH and Range-LSH, ProMIPS requires fewer candidate points to guarantee the accuracy benefiting from the derived effective searching conditions. In addition, the process of Quick-Probe determines a certain searching range in the projected space, which avoids testing each returned point to reduce the CPU time. With respect to H2-ALSH, it's more complex to count points' frequencies to fetch the candidate points in more hash tables compared to directly scanning points in hash tables for candidate points in Range-LSH. Therefore, it takes more CPU's running time in H2-ALSH.

Furthermore, we also evaluate the total time to verify the efficiency. Due to the space limits, we only show the experimental results on Netflix and Yahoo in Fig.~\ref{Total_Time}. In the whole searching process, a large portion of the time consumption comes from reading data from disks. Since ProMIPS performs the best on page access, it obtains the superior performance on total time.
\vspace{-15pt}
\subsection{Impact of $c$ and $p$}
\vspace{-5pt}
Since ProMIPS guarantees $c$-$k$-MIP search in accuracy, we vary the approximation ratio $c$ and the guaranteed probability $p$ to evaluate how the performances of ProMIPS vary with $c$ and $p$. We test overall ratio, recall, page access, CPU time and total time on four datasets. Due to the space limits, we only show the results on the overall ratio and page access to demonstrate our method's accuracy and efficiency. The recall and running time show similar trends with the overall ratio and page access, respectively. The experimental results are reported in Fig.~\ref{Imp_c} and Fig.~\ref{Imp_p}.

In Fig.~\ref{Imp_c}, the overall ratio decreases as $c$ decreases. This is because a smaller $c$ leads to a smaller range according to the searching conditions and fewer candidate points are selected, which leads to a lower accuracy. Although the overall ratio decreases, it's still larger than the given approximation ratio $c$. It demonstrates that ProMIPS can guarantee $c$-$k$-MIP search in accuracy. In Fig.~\ref{Imp_c}, a larger overall ratio leads to more page accesses, which shows that ProMIPS enjoys a better trade-off between the accuracy and efficiency.

In Fig.~\ref{Imp_p}, it shows that a higher probability leads to a higher overall ratio. This is because a higher $p$ leads to a larger searching range containing more candidate points. But more candidate points also incur more page accesses. Although we can obtain a higher overall ratio when $p=0.9$, it incurs much more page accesses at the same time. It demonstrates that the increasing rate of accuracy is lower than the decreasing rate of efficiency as $p$ increases.

\begin{figure}
\hspace{-10pt}
\subfigure[Netflix]{
\includegraphics[width=.25\textwidth]{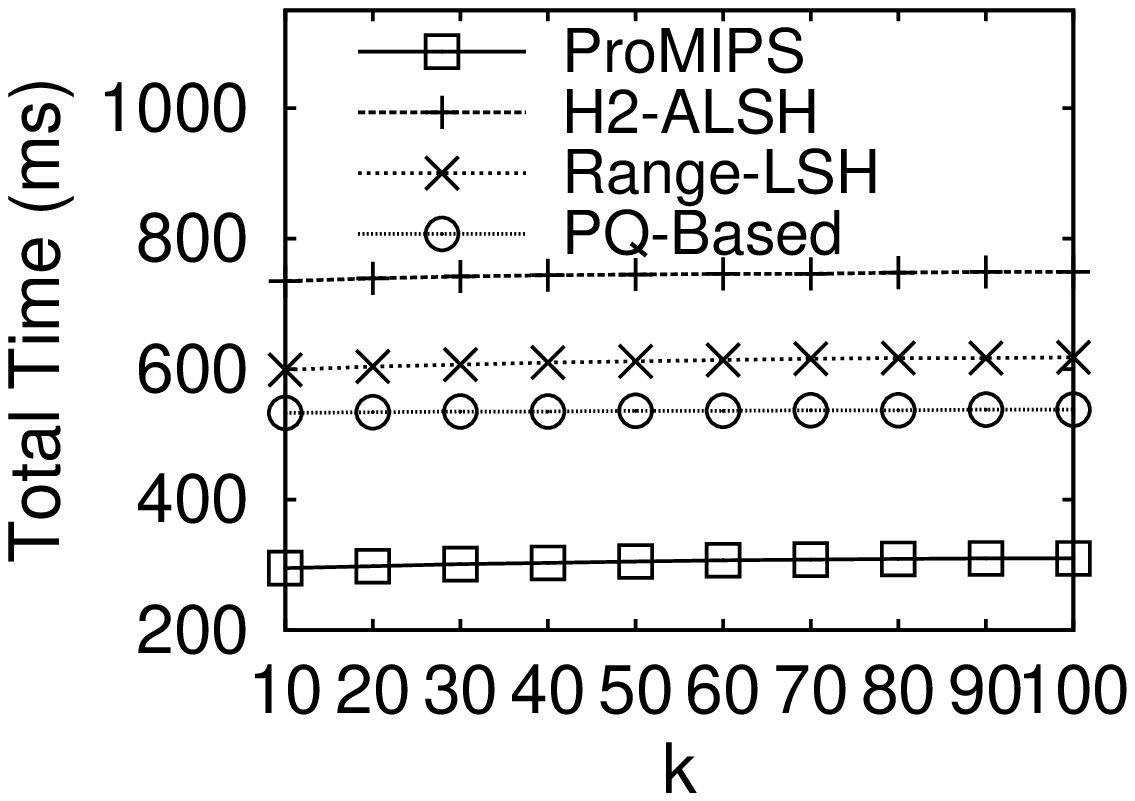}}
\hspace{-10pt}
\subfigure[Yahoo]{
\includegraphics[width=.25\textwidth]{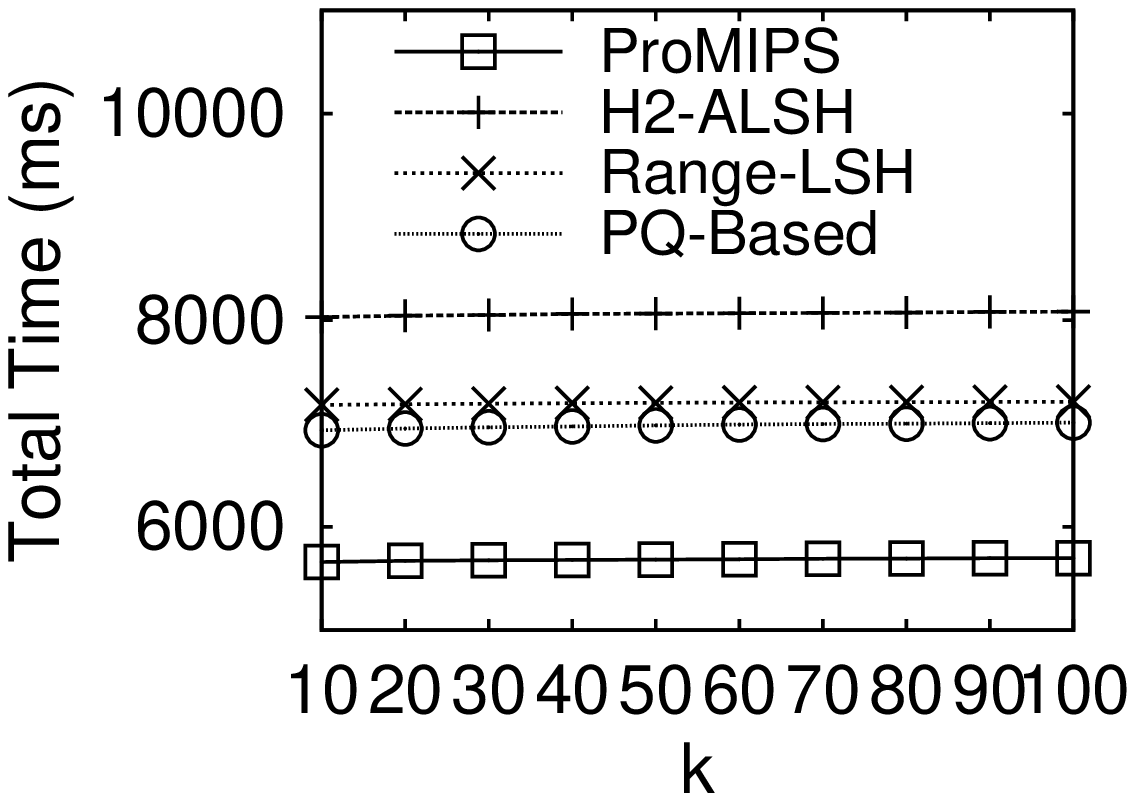}}
\caption{Total Time}
\label{Total_Time}
\vspace{-5pt}
\end{figure}

\begin{figure}
\hspace{-10pt}
\subfigure[Overall Ratio]{
\includegraphics[width=.25\textwidth]{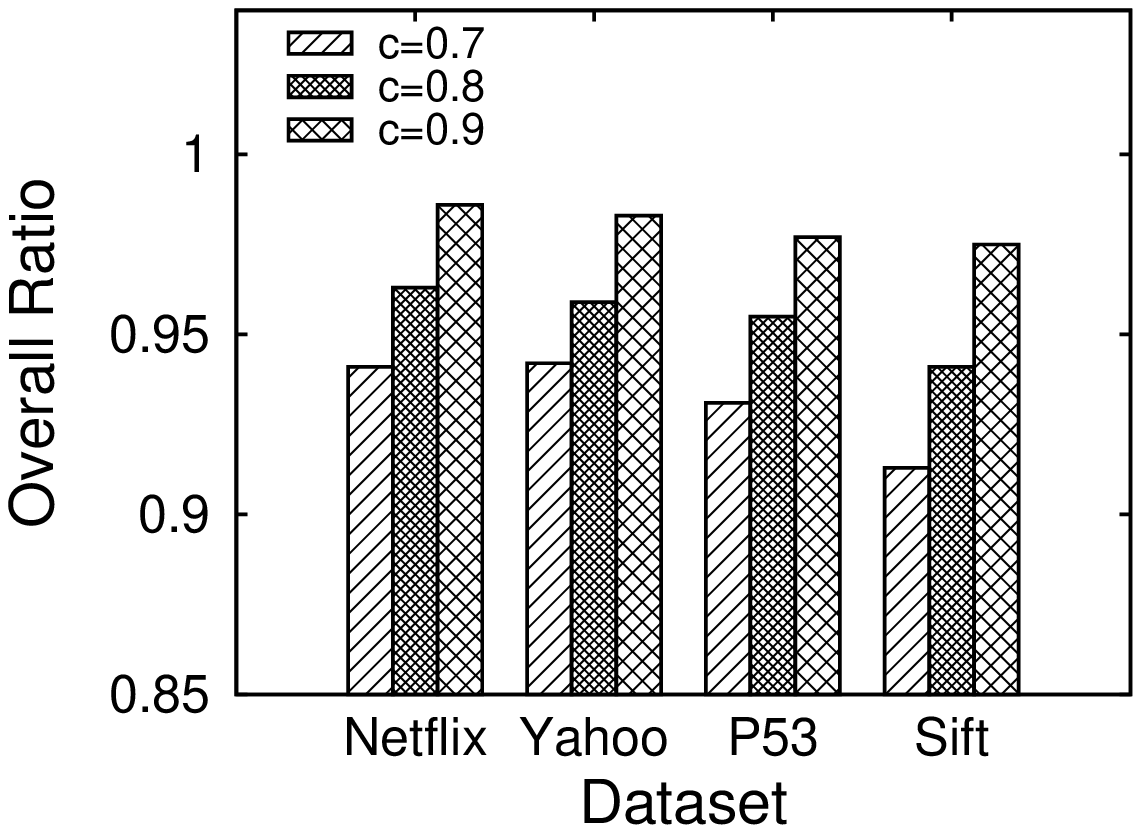}}
\hspace{-10pt}
\subfigure[Page Access]{
\includegraphics[width=.25\textwidth]{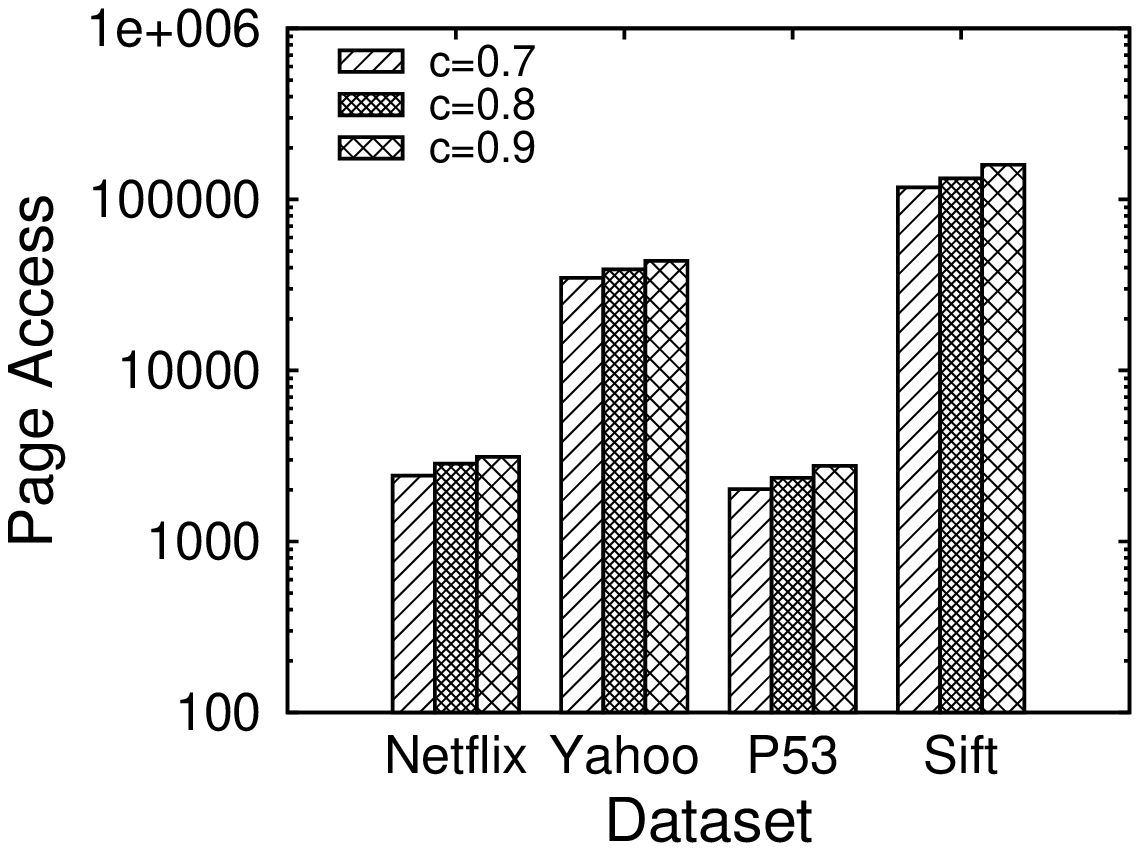}}
\caption{Impact of $c$}
\label{Imp_c}
\vspace{-12pt}
\end{figure}

\begin{figure}
\hspace{-10pt}
\subfigure[Overall Ratio]{
\includegraphics[width=.25\textwidth]{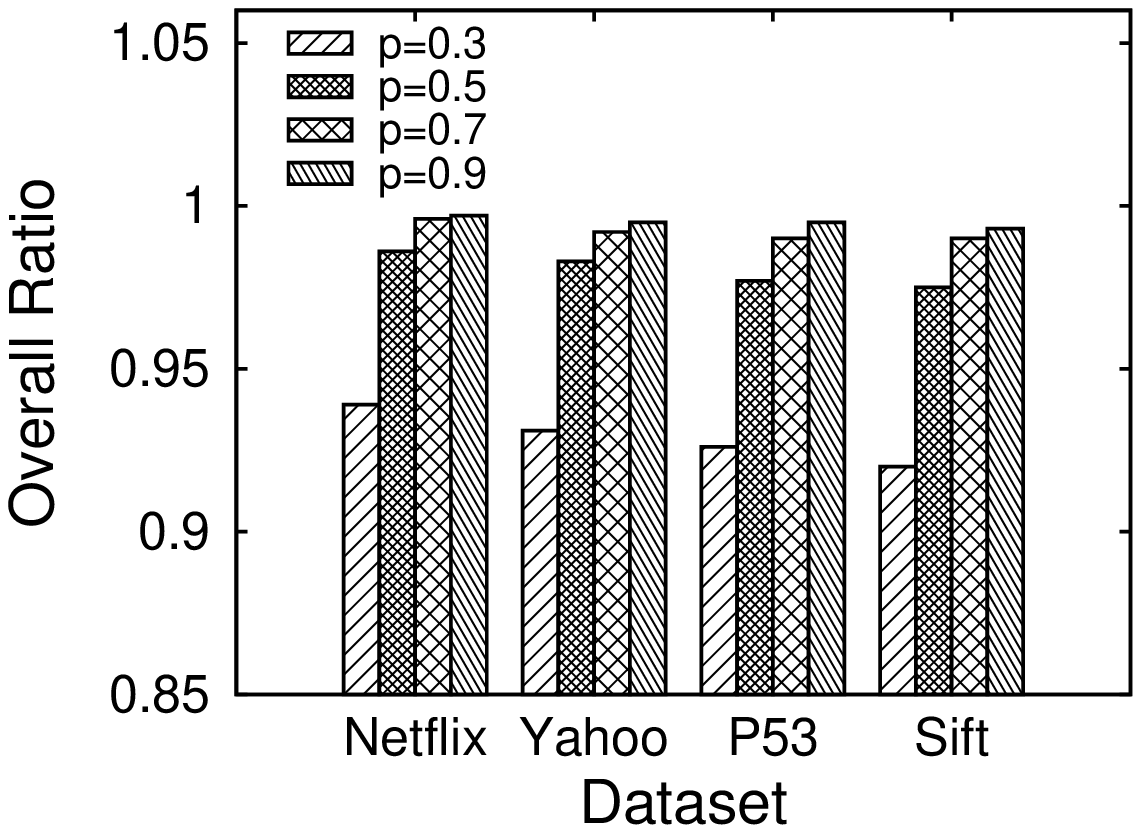}}
\hspace{-10pt}
\subfigure[Page Access]{
\includegraphics[width=.25\textwidth]{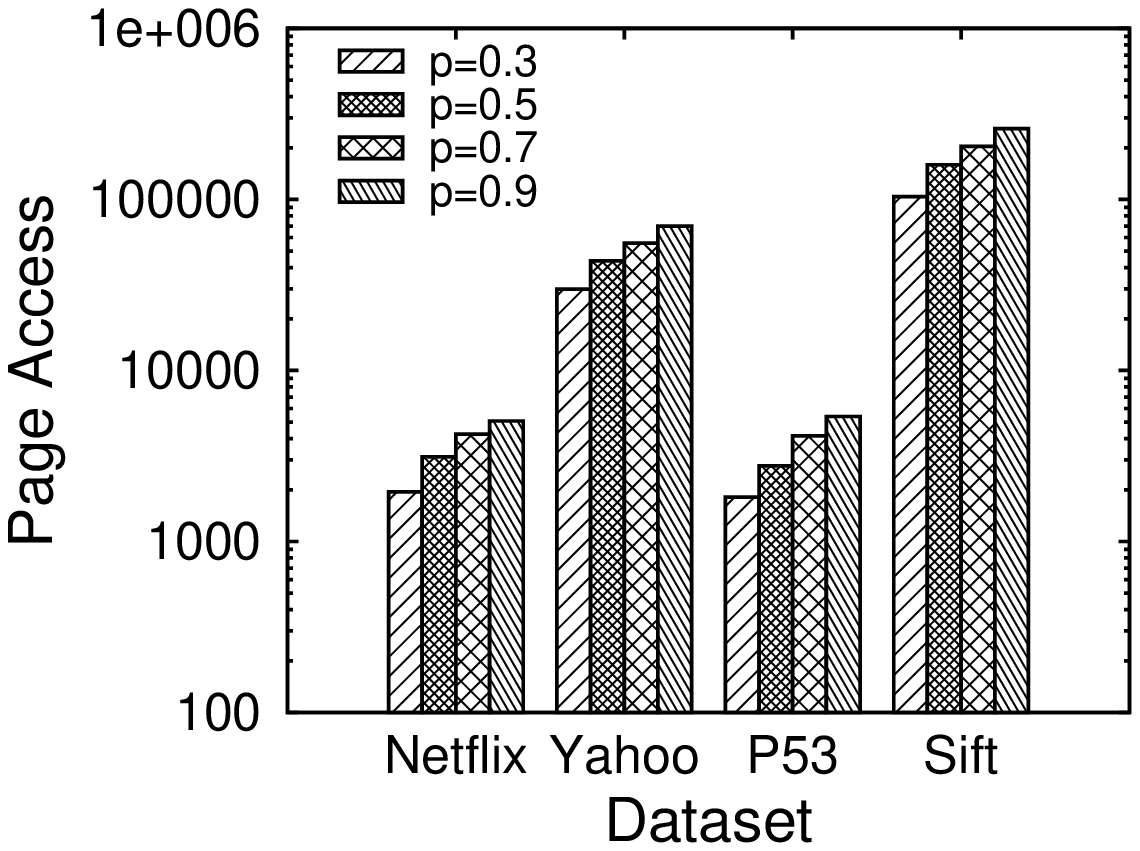}}
\caption{Impact of $p$}
\label{Imp_p}
\end{figure}

\section{Related Work}\label{sec:relatedwork}
In recent years, MIP search problem has received widespread attention and various types of methods have been proposed to solve both exact and approximate MIP search problems. In the beginning, some tree-based searching methods~\cite{DBLP:conf/kdd/RamG12, DBLP:conf/cikm/KoenigsteinRS12, DBLP:conf/sdm/CurtinGR13, DBLP:journals/sadm/CurtinR14} are presented for the exact MIP search problem. In addition, several methods based on linear search~\cite{DBLP:conf/sigmod/TeflioudiGM15, DBLP:conf/sigmod/LiCYM17, DBLP:journals/tods/TeflioudiG17} are also proposed. However, these methods suffer from the curse of dimensionality and their performances will degrade sharply when the feature dimension is high (more than 20)~\cite{DBLP:conf/kdd/HuangMFFT18, DBLP:conf/nips/YanLDCC18}.

To address the MIP search problem in high-dimensional space, there exists a line of research on approximate solutions by trading off the accuracy and efficiency. Since inner product doesn't satisfy some important metric properties such as non-negativity and triangle inequality, it's not a metric measurement. Existing methods for metric measurements~\cite{DBLP:journals/tkde/LiZSWLZL20, DBLP:journals/dase/LiL18, DBLP:conf/icde/ZhangOT04}, such as Locality Sensitive Hashing (LSH)~\cite{DBLP:journals/pvldb/HuangFZFN15} and some quantization-based methods~\cite{DBLP:journals/pvldb/WeiWWLZ0C20, DBLP:conf/sigmod/YangLFW20, DBLP:conf/cvpr/KalantidisA14}, can't be applied to MIP search problem. In addition, some methods proposed for a class of measurements such as Bregman distance~\cite{DBLP:journals/corr/abs-2006-00227} can't be employed. For this reason, most existing methods employ asymmetric (data points and query point are transformed in different manners) or symmetric (data points and query point are transformed in the same manner) transformations to convert a MIP search problem into a Nearest Neighbor (NN) search problem (called MIPS-NNS reduction) or a Maximum Cosine-similarity (MC) search problem (called MIPS-MCS reduction)~\cite{DBLP:conf/kdd/HuangMFFT18, DBLP:conf/nips/YanLDCC18, DBLP:conf/recsys/BachrachFGKKNP14, DBLP:conf/nips/Shrivastava014, DBLP:conf/icml/NeyshaburS15, DBLP:conf/uai/Shrivastava015}. Benefiting from these transformations, the order of MIP points can be preserved by the order of NN/MC points as much as possible, and the traditional metric search methods represented by LSH can be applied. These methods are considered as transformation-based methods and they are introduced as follows.

In L2-ALSH~\cite{DBLP:conf/nips/Shrivastava014} and Sign-ALSH~\cite{DBLP:conf/uai/Shrivastava015}, the MIP search problem is respectively converted into an NN search problem or an MC search problem by various asymmetric transformations, and the NN and MC search problems are solved by E2LSH~\cite{DBLP:conf/compgeom/DatarIIM04} and SimHash~\cite{DBLP:conf/stoc/Charikar02}, respectively. Nonetheless, they both introduce transformation errors affecting the accuracy. Besides, L2-ALSH leads to distortion errors after the transformation, which indicates that the Euclidean distance between most data points and the query point will be close to each other~\cite{DBLP:conf/kdd/HuangMFFT18}, and the efficiency will decrease. To avoid the transformation errors, an exact asymmetric transformation based solution named X-BOX is proposed. It takes advantage of the MIPS-NNS reduction and solves the NN search problem by PCA-tree, but its transformation also causes distortion errors. In addition to the aforementioned asymmetric solutions, Simple-LSH~\cite{DBLP:conf/icml/NeyshaburS15} employs a symmetric transformation for a MIPS-MCS reduction. Nevertheless, it suffers from long tails in the 2-norm distribution of real datasets~\cite{DBLP:conf/nips/YanLDCC18}.

Recently, two LSH-based methods, named H2-ALSH~\cite{DBLP:conf/kdd/HuangMFFT18} and Norm Ranging-LSH~\cite{DBLP:conf/nips/YanLDCC18} are devised. H2-ALSH proposes an asymmetric transformation without transformation errors named QNF transformation to convert MIP search problem into NN search problem. Furthermore, to reduce the distortion errors for the higher efficiency, a novel homocentric hypersphere partition strategy is designed. Norm-ranging LSH partitions the whole dataset into several subsets, where the searching process is performed by several independent indexes, to solve the excessive normalization problem caused by the long tails. Nevertheless, these methods require a large number of hash tables or long hash codes to ensure the accuracy, which takes up lots of pre-processing overheads. In this paper, we choose these two advanced methods as the benchmark methods.

There is also a plethora of data-dependent methods~\cite{DBLP:conf/iccv/ShenLZYS15, DBLP:conf/aaai/FraccaroPW16, DBLP:conf/ijcnn/KeivaniSR17, DBLP:conf/nips/YuHLD17, DBLP:conf/nips/MorozovB18, DBLP:conf/aaai/LiuWM19, DBLP:conf/aistats/DingYH19, DBLP:journals/corr/abs-1903-10391, DBLP:journals/corr/abs-1908-08656, DBLP:journals/corr/abs-1908-10396}, which are dedicated to the MIP search problem recently. These methods require learning-based techniques in the preprocess, which is difficult to maintain when large volumes of data are being updated. More importantly, they are not tailored to our concerned c-AMIP search problem with the probability guarantee in accuracy.

\section{Conclusion}\label{sec:conclusion}
In this paper, we address the important issue of $c$-AMIP search on high-dimensional and large-scale datasets by introducing an efficient method with a lightweight index. In our method, we employ 2-stable random projections to reduce the high-dimensional $c$-AMIP search problem to a low-dimensional search problem. With two derived searching conditions and the proposed Quick-Probe, our method can efficiently guarantee $c$-AMIP search in accuracy with arbitrary probabilities. In addition, to accelerate the searching process, we utilize the lightweight iDistance as the index to perform the range search in the low-dimensional space. Experimental results on four real datasets demonstrate that our method requires less pre-processing cost and provides $c$-AMIP results with a probability guarantee in accuracy efficiently.

\section*{Acknowledgment}
The work is supported by the National Key Research \& Development Program of China (No. 2018YFB1003400), the National Natural Science Foundation of China (Nos. 62072083 and U1811261) and Liaoning Revitalization Talents Program (XLYC1807158). Yang Song is supported by the Chinese Scholarship Council.

\bibliographystyle{abbrv}
\bibliography{reference}

\end{document}